\documentclass[12pt,onecolumn,draftclsnofoot]{IEEEtran} %!PN

\usepackage{amssymb}
\usepackage{graphicx, subfigure}
\usepackage{url, cite}

\newtheorem{theorem}{Theorem}

\newtheorem{subsec:coding}{subsec:coding}
\newtheorem{fact}{Fact}

\newtheorem{lemma}{Lemma}

\newcommand{\ls}[1]  %% 1 in brackets means \ls takes 1 argument
   {\dimen0=\fontdimen6\the=#1\dimen0
    \advance\lineskip.5\fontdimen5\the\lineskip-\dimen0
    \lineskiplimit=.9\lineskip
    \baselineskip=\lineskip
    \advance\baselineskip\dimen0
    \normallineskip\lineskip
    \normallineskiplimit\lineskiplimit
    \normalbaselineskip\baselineskip
    \ignorespaces
   }

\begin{document}

\title{On Scalable Video Streaming over Cognitive Radio Cellular and Ad Hoc Networks}

\author{
\authorblockN{Donglin Hu \ \ and \ \ Shiwen Mao \\}
\authorblockA{Department of Electrical and Computer Engineering  \\
Auburn University, Auburn, AL, USA} 
}

\maketitle

\pagestyle{plain}\thispagestyle{plain}

%\ls{0.89}

\begin{abstract}
Video content delivery over wireless networks is expected to grow drastically in the coming years. In this paper, we investigate the challenging problem of video over cognitive radio (CR) networks. Although having high potential, this problem brings about a new level of technical challenges. After reviewing related work, we first address the problem of video over infrastructure-based CR networks, and then extend the problem to video over non-infrastructure-based ad hoc CR networks. We present formulations of cross-layer optimization problems as well as effective algorithms to solving the problems. The proposed algorithms are analyzed with respect to their optimality and validate with simulations.  
\end{abstract}

%-------------------------------------------------------------------
% Introduction
\section{Introduction}
Video content delivery over wireless networks is expected to grow drastically in the coming years. The compelling need for ubiquitous video content access will significantly stress the capacity of existing and future wireless networks. To meet this critical demand, the Cognitive Radio (CR) technology provides an effective solution that can effectively exploit co-deployed networks and aggregate underutilized spectrum for future video-aware wireless networks.

The high potential of CRs has attracted substantial interest. The mainstream CR research has focused on developing effective spectrum sensing and access techniques (eg., see~\cite{Zhao07a, Zhao09}). Although considerable advances have been achieved, the important problem of guaranteeing application performance has not been well studied.  We find video streaming can make excellent use of the enhanced spectrum efficiency in CR networks.  Unlike data, where each bit should be delivered, video is loss-tolerant and rate-adaptive~\cite{Mao05-wcomm, Mao05-icc}.  They are highly suited for CR networks, where the available bandwidth depends on primary user transmission behavior.  Graceful degradation of video quality can be achieved as spectrum opportunities evolve over time.

CR is an evolving concept with various network models and levels of cognitive functionality~\cite{Zhao07a, Zhao09}. IEEE 802.22 Wireless Regional Area Networks (WRAN) is the first CR standard for reforming broadcast TV bands, where a base station (BS) controls medium access for customer-premises equipments (CPEs)~\cite{802.22STD}. Therefore, we first consider multicasting scalable videos in such an infrastructure-based CR network. The spectrum consists of multiple channels, each allocated to a primary network. The CR network is co-located with the primary networks, where a CR BS seeks spectrum opportunities for multicasting multiple video streams, each to a group of secondary subscribers. The problem is to exploit spectrum opportunities for minimizing video distortion, while keeping the collision rate with primary users below a prescribed threshold. We consider scalable video coding, such as fine-grained-scalability (FGS) and medium grain scalable (MGS) videos~\cite{Schaar03, Wien07}. We model the problem of CR video multicast over the licensed channels as a mixed integer nonlinear programming (MINLP) problem, and then develop a sequential fixing algorithm and a greedy algorithm to solve the MINLP, while the latter has a low computational complexity and a proved optimality gap~\cite{Hu09IC}.

We then tackle the problem of video over multi-hop CR networks, e.g., a wireless mesh network with CR-enabled nodes. This problem is more challenging than the problem above due to the lack of infrastructure support. We assume each secondary user is equipped with two transceivers. To model and guarantee end-to-end video performance, we adopt the amplify-and-forward approach for video data transmission, which is well-studied in the context of cooperative communications~\cite{Laneman04}. This is equivalent to setting up a ``virtual tunnel'' through a multi-hop multi-channel path. The challenging problem, however, is how to set up the virtual tunnels, while the available channels at each relay evolve over time due to primary user transmissions. 
The formulated MINLP problem is first solved using a centralized sequential fixing algorithm, which provides upper and lower bounds for the achievable video quality. We then apply dual decomposition to develop a distributed algorithm and prove its optimality as well as the convergence condition~\cite{Hu10TW}.  

The rest of the paper is organized as follows. We review related work in Section~\ref{sec:cr_video_work} and present preliminaries in Section~\ref{sec:cr_video_sys}. We examine video over infrastructure-based CR networks in Section~\ref{sec:cr_video_mcast} and over multi-hop CR networks in Section~\ref{sec:cr_video_mhop}. We concludes the paper in Section~\ref{sec:cr_video_conc} with a discussion of open problems.

%-------------------------------------------------------------------
% Background and Related Work
\section{Background and Related Work}\label{sec:cr_video_work}
The high potential of CRs has attracted considerable interest form both industry, government and academia~\cite{Akyildiz06, Zhao07a}. The mainstream CR research has been focused on spectrum sensing and dynamic spectrum access issues. For example, the impact of spectrum sensing errors on the design of spectrum access schemes has been addressed in several papers~\cite{Chen08,Urgaonkar09,Shu09}. The approach of iteratively sensing a selected subset of available channels has been developed in the design of CR MAC protocols~\cite{Su08,Zhao08}. The optimal trade-off between the two kinds of sensing errors is investigated comprehensively and addressed in depth in~\cite{Chen08}.

The important issue of QoS provisioning in CR networks has been studied only in a few papers~\cite{Su08,Fattahi07}, where the objective is still focused on the so-called ``network-centric'' metrics such as maximum throughput and delay~\cite{Su08,Urgaonkar09}. In~\cite{Urgaonkar09},  an interesting delay throughput trade-off for a multi-cell cognitive radio network is derived, while the goal of primary user protection is achieved by stabilizing a virtual ``collision queue''. In~\cite{Fattahi07}, a game-theoretic framework is described for resource allocation for multimedia transmissions in spectrum agile wireless networks. In this interesting work, each wireless station participates in a resource management game, which is coordinated by a network moderator. A mechanism-based resource management scheme determines the amount of transmission opportunities to be allocated to various users on different frequency bands such that certain global system metrics are optimized. 

%Recently, the end-to-end user-perceived video quality experienced by CR users becomes more and more important. 
The problem of video over CR networks has been addressed only in a few recent papers. In~\cite{Shiang08}, a priority virtual queue model is adopted for wireless CR users to select channel and maximize video qualities. In~\cite{Luo11}, the impact of system parameters residing in different network layers are jointly considered to achieve the best possible video quality for CR users. The problem is formulated as a Mini-Max problem and solved with a dynamic programming approach. In~\cite{Ali09}, Ali and Yu jointly optimize video parameter with spectrum sensing and access strategy. A rate-distortion model is adopted to optimize the intra-mode selection and source-channel rate with a partially observable Markov decision process (POMDP) formulation. In~\cite{Guan11}, video encoding rate, power control, relay selection and channel allocation are jointly considered for video over cooperative CR networks. The problem is formulated as a mixed-integer nonlinear problem and solved by a solution algorithm based on a combination of the branch and bound framework and convex relaxation techniques. 

Video multicast, as one of the most important multimedia services, has attracted considerable interest from the research community. Layered video multicast has been researched in the mobile ad hoc networks~\cite{Mao06,Wei07} and infrastructure-based wireless networks~\cite{Deb08,Schaar03}. A greedy algorithm is presented in~\cite{Deb08} for layered video multicast in WiMAX networks with a proven optimality gap.

A few recent works~\cite{Hou07,Hou08,Feng09} have studied multi-hop CR networks. The authors formulate cross-layer optimization problem considering factors from the PHY up to the transport layer. The dual decomposition technique~\cite{Palomar06, Bertsekas99} is adopted to develop distributed algorithm. We choose similar methodology in our work and apply it to the more challenging problem of real-time video streaming. 

%-------------------------------------------------------------------
% System Model and Preliminaries
\section{System Model and Preliminaries}\label{sec:cr_video_sys}
%\input{chapter3/system.tex}
%-------------------------------------------------------------------
\subsection{Primary Network}
%-------------------------------------------------------------------

We consider a spectrum band consisting of $M$ orthogonal channels with identical bandwidth~\cite{Jia08,Hu11MC}. We assume that the $M$ channels are allocated to $K$ primary networks, which cover different service areas. A primary network can use any of the $M$ channels without interfering with other primary networks.  We further assume that the primary systems use a synchronous slot structure as in prior work~\cite{Zhao07a, Su08}.  Due to primary user transmissions, the occupancy of each channel evolves following a discrete-time Markov process, as validated by recent measurement studies~\cite{Zhao07a, Su08, Motamedi07}.

In primary network $k$, the status of channel $m$ in time slot $t$ is denoted by $S_m^k(t)$ with idle (i.e., $S_m^k(t)=0$) and busy (i.e., $S_m^k(t)=1$) states.  Let $\lambda_m^k$ and $\mu_m^k$ be the transition probability of remaining in state $0$ and that from state $1$ to $0$, respectively, for channel $m$ in primary network $k$. The utilization of channel $m$ in primary network $k$, denoted by $\eta_m^k = \Pr(S_m^k=1)$, is
\begin{eqnarray}\label{eq:ProbBusy1}
  \eta_m^k = \lim_{T \to \infty} \frac{1}{T} 
             \mbox{$\sum_{t=1}^T$} S_m^k(t) 
           = \frac{1-\lambda_m^k}{1-\lambda_m^k+\mu_m^k}.
\end{eqnarray}

Note that in infrastructure-based CR networks and cooperative CR networks, we assume there is only one $K=1$ primary network. In infrastructure-based CR networks introduced in the section \ref{sec:cr_video_mcast}, we adopt $N$ as the number of licensed channels since $M$ is denoted as the number of modulation-coding schemes. 

%-------------------------------------------------------------------
\subsection{Infrastructure-based CR Networks}
%-------------------------------------------------------------------
As shown in Fig.~\ref{fig:ntwarch}, 
we consider a CR base station multicasts $G$ real-time videos to $G$ multicast groups, each of which have $N_g$ users, $g=1, 2, \cdots, G$.  The base station seeks spectrum opportunities in the $N$ channels to serve CR users.  In each time slot $t$, the base station selects a set of channels $\mathcal{A}_1(t)$ to sense and a set of channels $\mathcal{A}_2(t)$ to access. Without loss of generality, the base station has $|\mathcal{A}_1(t)|$ transceivers such that it can sense $|\mathcal{A}_1(t)|$ channels simultaneously. Note that a time slot and channel combination, termed a {\em tile}, is the minimum unit for resource allocation.

%---------------------------------------------------------------
\begin{figure} [!t]%[ht]
\centering
\includegraphics[width=4.3in]{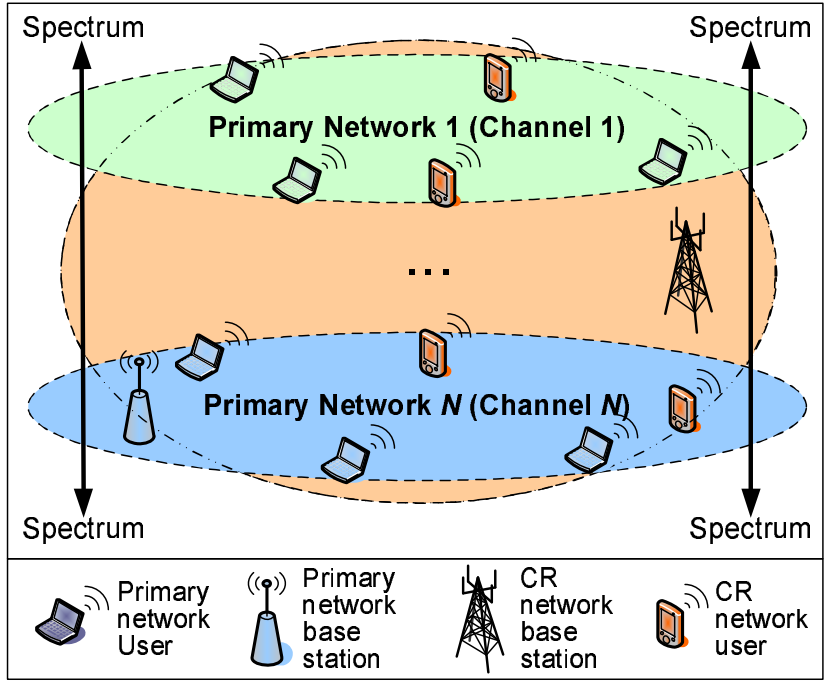}
\caption{An infrastructure-based CR network collocated with $N$ primary networks.}
\label{fig:ntwarch}
\end{figure}
%---------------------------------------------------------------

We adopt the same time-slot structure as in~\cite{Zhao07a, Zhao08}. 
, which is illustrated in Fig.~\ref{fig:SlotStructure1}.  
At the beginning of each time slot, the base station senses channels in $\mathcal{A}_1(t)$ and then chooses a set of available channels for opportunistic transmissions based on sensing results. After a successful transmission, the base station will receive an ACK from the user with the highest SNR in the target multicast group. Without loss of generality, we assume that each CR network user can access all the available channels with the channel bonding/aggregation techniques~\cite{Corderio06, Mahmoud09}. 
%~\cite{Zhao07a, Berthold05, Tang05}. 
%When there are multiple licensed channels, we first consider the case where the channel bonding/aggregation techniques are used by the transmitters and CR users [13], [14].

%-------------------------------------------------------------------
\begin{figure} [!t]
\centering
\includegraphics[width=4.3in]{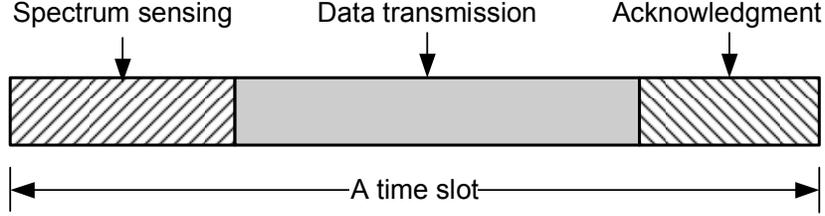}
\caption{The structure of a time slot.}
\label{fig:SlotStructure1}
\end{figure}
%-------------------------------------------------------------------

%-------------------------------------------------------------------
\subsection{Multi-hop CR networks}
%-------------------------------------------------------------------

As shown in Fig.~\ref{fig:Scenario}, we also consider a multi-hop CR network that is co-located with the primary networks, within which $\mathcal{S}$ real-time videos are streamed among $N$ CR nodes.  Let $\mathcal{U}^k$ denote the set of CR nodes that are located within the coverage of primary network $k$.  A video session $l$ may be relayed by multiple CR nodes if source $z_l$ is not a one-hop neighbor of destination $d_l$.  We assume a {\em common control channel} for the CR network~\cite{Su08}.  We also assume the timescale of the primary channel process (or, the time slot durations) is much larger than the broadcast delays on the control channel, such that feedbacks of channel information can be received at the source nodes in a timely manner. 

%-------------------------------------------------------------------
\begin{figure} [!t]
\centering
\includegraphics[width=4.5in]{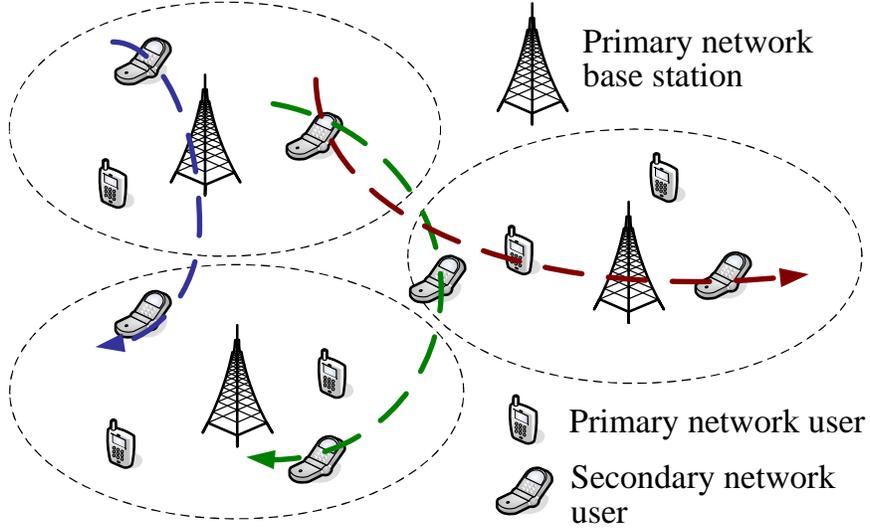}
\caption{Illustration of the multi-hop video CR network architecture.}
\label{fig:Scenario}
\end{figure}
%-------------------------------------------------------------------

The time slot structure is the same as that in infrastructure-based CR networks. In the sensing phase, one transceiver of a CR node is used to sense one of the $M$ channels, while the other is tuned to the control channel to exchange channel information with other CR users.  Each video source computes the optimal path selection and channel scheduling based on sensing results.  In the transmission phase, the channels assigned to a video session $l$ at each link along the path form a virtual ``tunnel'' connecting source $z_l$ and destination $d_l$. 
As illustrated in Fig.~\ref{fig:cutth},
each node can use one or more than one channels to communicate with other nodes using the channel bonding/aggregation techniques~\cite{Corderio06, Mahmoud09}.  When multiple channels are available on all the links along a path, multiple tunnels can be established and used simultaneously for a video session. In the acknowledgment phase, the destination sends ACK to the source for successfully received video packets through the same tunnel. 

We adopt amplify-and-forward for video transmission~\cite{Laneman04}.  During the transmission phase, one transceiver of the relay node receives video data from the upstream node on one channel, while the other transceiver of the relay node amplifies and forwards the data to the downstream node on a different, orthogonal channel.  There is no need to store video packets at the relay nodes.  Error detection/correction will be performed at the destination node.  As a result, we can transmit through the tunnel a block of video data with minimum delay and jitter in one time slot.

%-------------------------------------------------------------------
\begin{figure} [!t]
\centering
\includegraphics[width=5.0in]{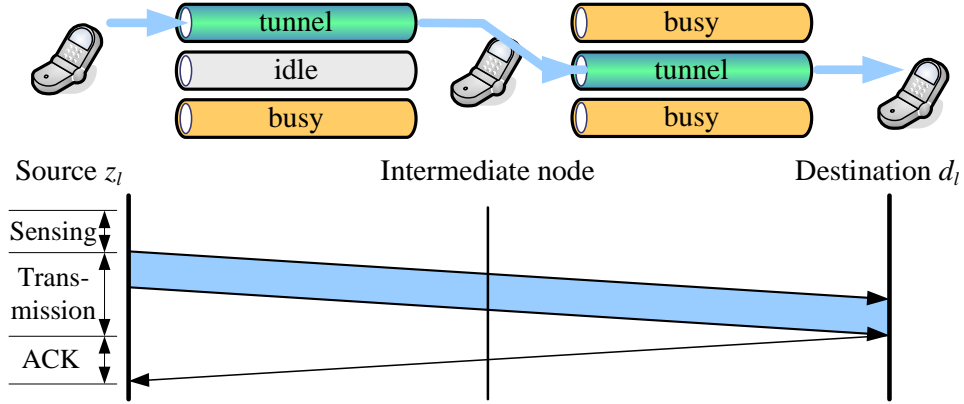}
\caption{The cut-through switching model for video data.}
\label{fig:cutth}
\end{figure}
%-------------------------------------------------------------------

%-------------------------------------------------------------------
\subsection{Spectrum Sensing \label{subsec:ssen}}
%-------------------------------------------------------------------

Two types of sensing errors may occur during the sensing process. A {\em false alarm} may lead to waste a spectrum opportunity and a {\em miss detection} may causes collision with primary users. In a multi-hop CR network, the sensing results from various users may be different. Denote $H_0$ as the hypothesis that channel $m$ in primary network $k$ is idle, and $H_1$ the hypothesis that channel $m$ in primary network $k$ is busy in time slot $t$.  The conditional probability that channel $m$ is available in primary network $k$, denoted by $a_m^k(t)$, can be derived as,  
%\begin{figure*}
%\begin{eqnarray}\label{eq:AvailProb}
%	& & \hspace{-0.2in} a_m^k(t) = \Pr(H_0|W_i^m=\theta_i^m, \; i \in \mathcal{U}_m^k,
%	          \pi_m^k) 
%	= \frac{\Pr(W_i^m=\theta_i^m,i \in \mathcal{U}_m^k|H_0,\pi_m^k)\Pr(H_0|\pi_m^k)}{\sum_{s\in\{0,1\}} \Pr(W_i^m=\theta_i^m,i \in \mathcal{U}_m^k|H_s,\pi_m^k)\Pr(H_s|\pi_m^k)} \nonumber\\
%	& & \hspace{-0.4in} = \frac {\Pr(H_0|\pi_m^k)\prod_{i \in \mathcal{U}_m^k}\Pr(W_i^m=\theta_i^m|H_0,\pi_m^k)}{\sum_{s\in\{0,1\}} \Pr(H_s|\pi_m^k)\prod_{i \in \mathcal{U}_m^k}\Pr(W_i^m=\theta_i^m|H_s,\pi_m^k)} 
%	= \frac {\Pr(H_0|\pi_m^k)\prod_{i \in \mathcal{U}_m^k}\Pr(W_i^m=\theta_i^m|H_0)}{\sum_{s\in\{0,1\}} \Pr(H_s|\pi_m^k)\prod_{i \in \mathcal{U}_m^k}\Pr(W_i^m=\theta_i^m|H_s)} \nonumber \\
%	& & \hspace{-0.4in} = \left[ 1+\frac{\Pr(H_1|\pi_m^k)}{\Pr(H_0|\pi_m^k)}\prod_{i \in \mathcal{U}_m^k} 
%	\frac{\Pr(W_i^m=\theta_i^m|H_1)}{\Pr(W_i^m=\theta_i^m|H_0)} \right]^{-1} 
%	= \left[ 1 + \left( \varphi_m^k \right)^{u_m^k} \left( \phi_m^k
%	         \right)^{|\mathcal{U}_m^k|-u_m^k}
%	      \frac{\Pr(H_1|\pi_m^k)}{\Pr(H_0|\pi_m^k)} \right]^{-1}.
%\end{eqnarray}
%\end{figure*}
\begin{eqnarray}\label{eq:AvailProb1}
	a_m^k(t) &=& \Pr(H_0|W_i^m=\theta_i^m, \; i \in \mathcal{U}_m^k, \pi_m^k) 
	         \nonumber \\ 
%	&=& \frac{\Pr(W_i^m=\theta_i^m,i \in \mathcal{U}_m^k|H_0, \pi_m^k) \Pr(H_0|\pi_m^k)}{\sum_{s\in\{0,1\}} \Pr(W_i^m=\theta_i^m,i \in \mathcal{U}_m^k|H_s,\pi_m^k)\Pr(H_s|\pi_m^k)} \nonumber \\
%	&=& \frac {\Pr(H_0|\pi_m^k)\prod_{i \in \mathcal{U}_m^k}\Pr(W_i^m=\theta_i^m|H_0,\pi_m^k)}{\sum_{s\in\{0,1\}} \Pr(H_s|\pi_m^k)\prod_{i \in \mathcal{U}_m^k}\Pr(W_i^m=\theta_i^m|H_s,\pi_m^k)} \nonumber \\
%	&=& \frac {\Pr(H_0|\pi_m^k)\prod_{i \in \mathcal{U}_m^k}\Pr(W_i^m=\theta_i^m|H_0)}{\sum_{s\in\{0,1\}} \Pr(H_s|\pi_m^k)\prod_{i \in \mathcal{U}_m^k}\Pr(W_i^m=\theta_i^m|H_s)} \nonumber \\
%	&=& \left[ 1+\frac{\Pr(H_1|\pi_m^k)}{\Pr(H_0|\pi_m^k)}\prod_{i \in \mathcal{U}_m^k} 
%	\frac{\Pr(W_i^m=\theta_i^m|H_1)}{\Pr(W_i^m=\theta_i^m|H_0)} \right]^{-1}  \nonumber \\
	&=& \left[ 1 + \left( \varphi_m^k \right)^{u_m^k} \left( \phi_m^k
	         \right)^{|\mathcal{U}_m^k|-u_m^k}
	      \frac{\Pr(H_1|\pi_m^k)}{\Pr(H_0|\pi_m^k)} \right]^{-1}.
\end{eqnarray}
where $\theta_i^m$ represents a specific sensing result (0 or 1), $\mathcal{U}_m^k$ is the subset of users in $\mathcal{U}^k$ (i.e., the set of CR nodes that are located within the coverage of primary network $k$) that sense channel $m$, $u_m^k$ is the number of users in $\mathcal{U}_m^k$ observing channel $m$ is idle, $\pi_m^k$ represents the history of channel $m$ in primary network $k$,
%\footnote{$\pi_m^k$ represents the availability of channel $m$ in primary network $k$ in the previous time slot.  If the channel was used in that time slot, $\pi_m^k$ can be readily determined as 0 or 1, since the channel state was known (i.e., with or without ACKs). Otherwise, $\pi_m^k$ can be estimated in the form of $a_m^k (t-1)$ as in~(\ref{eq:AvailProb}).} 
and $\varphi_m^k$ and $\phi_m^k$ are defined as: 
\begin{eqnarray}\label{eq:DefVarPhi}
\left\{
\begin{array}{ll}
\varphi_m^k=\frac{P(W_i^m=0|H_1)}{P(W_i^m=0|H_0)}=\frac{\delta_m}{1-\epsilon_m}, &
	  \mbox{when } \theta_i^m=0 \\
\phi_m^k=\frac{P(W_i^m=1|H_1)}{P(W_i^m=1|H_0)}=\frac{1-\delta_m}{\epsilon_m}, & 
	  \mbox{when } \theta_i^m=1 .
\end{array}\right.
\end{eqnarray}
%The third equality is due to independent sensing processes.  The fourth equality is because sensing processes are independent of channel history.  
Based on the Markov chain channel model, we have (\ref{eq:OneStep}), which can be recursively expanded:
\begin{eqnarray} \label{eq:OneStep}
  \left\{ \begin{array}{l}
    \Pr(H_0|\pi_m^k) = \lambda_m^k a_m^k(t-1) + \mu_m^k \left[ 1-a_m^k(t-1) \right] \\
    \Pr(H_1|\pi_m^k) = 1-\Pr(H_0|\pi_m^k).  
  \end{array} \right. 
\end{eqnarray}  

%-------------------------------------------------------------------
\subsection{Video Performance Measure}\label{sec:VideoRate}
%-------------------------------------------------------------------

Both FGS and MGS videos are highly suited for dynamic CR networks.  With FGS or MGS coding, each video $l$ is encoded into one base layer with rate $R_l^b$ and one enhancement layer with rate $R_l^e$.  The total bit rate for video $l$ is
  $R_l=R_l^b+R_l^e$.

We consider peak-signal-noise-ratio (PSNR) (in dB) of reconstructed videos.  As in prior work~\cite{Schaar03}, the average PSNR of video $l$, denoted as $Q_l$, can be estimated as:
\begin{equation}\label{eq:RDfun}
  Q_l(R_l) = Q_l^b + \beta_l(R_l - R_l^b) = Q_l^0 + \beta_l R_l,
\end{equation}
where 
$Q_l^b$ is the resulting PSNR when the base layer is decoded alone, $\beta_l$ a constant depending on the video sequence and codec setting, and $Q_l^0 = Q_l^b - \beta_l R_l^b$.  We verified the model (\ref{eq:RDfun}) with several test video sequences using the MPEG-4 FGS codec and the H.264/SVC MGS codec and found it is highly accurate. 
%The results are omitted for brevity.

Due to the real-time nature, we assume that each {\em group of pictures} (GOP) must be delivered during the next GOP window, which consists of $N_G$ time slots.  Beyond that, overdue data from the current GOP will be useless and will be discarded. In infrastructure-based network, $G$ video stream are multicast to $G$ groups of CR user, so we choose the group index $g$ instead of video session index $l$.

%-------------------------------------------------------------------
% Video over infrastructure based CR networks
\section{Video over Infrastructure Based CR Networks}\label{sec:cr_video_mcast}
In this section, we examine the problem of video over infrastructure-based CR networks. We consider cross-layer design factors such as scalable video coding, spectrum sensing, opportunistic spectrum access, primary user protection, scheduling, error control and modulation. We propose efficient optimization and scheduling algorithms for highly competitive solutions, and prove the complexity and optimality bound of the proposed greedy algorithm.

%-------------------------------------------------------------------
\subsection{Network Model}
%-------------------------------------------------------------------

%-------------------------------------------------------------------
\subsubsection{Spectrum Access \label{subsec:osa}}
%-------------------------------------------------------------------

At the beginning of each time slot $t$, the CR BS senses the $M$ channels and compute $a_n(t)$ for each channel $n$. Based on spectrum sensing results, the base station determines which channels to access for video streaming. We adopt an opportunistic spectrum access approach, aiming to exploit unused spectrum while probabilistically bounding the interference to primary users.  

Let $\gamma_n \in (0, 1)$ be the {\em maximum allowed collision probability} with primary users on channel $n$, and $p^{tr}_n(t)$ the {\em transmission probability} on channel $n$ for the base station in time slot $t$.  The probability of collision caused by the base station should be kept below $\gamma_n$, i.e., $p^{tr}_n(t) \left[ 1 - a_n(t) \right] \le \gamma_n$.  In addition to primary user protection, another important objective is to exploit unused spectrum as much as possible.  The transmission probability can be determined by jointly considering both objectives, as
\begin{equation}\label{eq:ptran}
  p^{tr}_n(t) = \left\{ \begin{array} {ll}
         \min \left\{ 1, \frac{\gamma_n}{1-a_n(t)} \right\}, & \mbox{ if }\; 0 \leq a_n(t) < 1 \\
         1, & \mbox{ if }\; a_n(t) =1. 
                        \end{array} \right.                        
\end{equation}
If $p^{tr}_n(t)=1$, channel $n$ will be accessed deterministically. If $p^{tr}_n(t) = \gamma_n / [1-a_n(t)] <1$, channel $n$ will be accessed opportunistically with probability $p^{tr}_n(t)$.

%-------------------------------------------------------------------
\subsubsection{Modulation-Coding Schemes \label{subsec:mcs}}
%-------------------------------------------------------------------

At the PHY layer, we consider various modulation and channel coding combination schemes.  Without loss of generality, we assume several choices of modulation schemes, such as QPSK, 16-QAM and 64-QAM, combined with several choices of forward error correction (FEC) schemes, e.g., with rates 1/2, 2/3, and 3/4. We consider $M$ unique combinations of modulation and FEC schemes, termed {\em Modulation-Coding} (MC) schemes, in this paper.

Under the same channel condition, different MC schemes will achieve different data rates and symbol error rates. Adaptive modulation and channel coding allow us to exploit user channel variations to maximize video data rate under a given residual bit error rate constraint. When a user has a good channel, it should adopt an MC scheme that can support a higher data rate. Conversely, it should adopt a low-rate MC scheme when the channel condition is poor.  Let $\{MC_m\}_{m=1,\cdots,M}$ be the list of available MC schemes indexed according to their data rates in the increasing order.  We assume slow fading channels with coherence time larger than a time slot.  Each CR user measures its own channel and feedbacks measurements to the base station when its channel quality changes.  At the beginning of a time slot, the base station is able to collect the number $n_{g,m}$ of users in each multicast group $g$ who can successfully decode $MC_m$ signals for $m=1, 2, \cdots, M$. 

Since the base layer carries the most important data, the most reliable MC scheme $MC_{b(g)}$ should be used, where $b(g) = \max_i \{ i : n_{g,i} = N_g \}$, for all $g$.  Without loss of generality, we assume that the base layer is always transmitted using $MC_1$. If a user's channel is so poor that it cannot decode the $MC_1$ signal, we consider it disconnected from the CR network.  We further divide the enhancement layer into $M$ sub-layers, where sub-layer $m$ has rate $R_{g,m}^e$ and uses $MC_m$.  Assuming that $MC_m$ can carry $b_{g,m}$ bits of video $g$ in one tile, we denote the number of tiles for sub-layer $m$ of video $g$ as $l_{g,m} \geq 0$. We have 
\begin{equation}
R_g^e = \sum_{m=1}^M R_{g,m}^e = \sum_{m=1}^M b_{g,m} l_{g,m}.
\end{equation}

%-------------------------------------------------------------------
\subsubsection{Proportional Fair Allocation \label{subsec:fairness}}
%-------------------------------------------------------------------

%For data communications, {\em proportional fairness} is a widely adopted measure, which can be achieved by maximizing the sum of logarithms of user rates (i.e., utilities)~\cite{Kelly98}.  
Since we consider video quality in this paper, we define the utility for user $i$ in group $g$ as $U_{g,i} = \log Q_{g,i} = \log \left( Q^b_g+\beta_g R^e_g(i) \right)$, where $R^e_g(i)$ is the received enhancement layer rate of user $i$ in group $g$.

The total utility for group $g$ is $U_g = \sum_{i=1}^{N_g} U_{g,i}$.  Intuitively, a lower layer should use a lower (i.e., more reliable) MC scheme.  This is because if a lower layer is lost, a higher layer cannot be used at the decoder even if it is correctly received.  Considering the user classification based on their MC schemes, we can rewrite $U_g$ as follows~\cite{Deb08}:
\begin{equation}
  U_g = \sum_{k=1}^M(n_{g,k}-n_{g,k+1})\log \left( Q^b_g+\beta_g \sum_{m=1}^k R_{g,m}^e \right),
\end{equation}
where $n_{g,M+1}=0$.  The utility function of the entire CR video multicast system is 
\begin{equation}
U = \sum_{g=1}^G U_g.
\end{equation}
Maximizing $U$ will achieve {\em proportional fairness} among the video sessions~\cite{Kelly98}

%-------------------------------------------------------------------
\subsection{Optimized Video Multicast in CR Networks \label{sec:optimized}}
%-------------------------------------------------------------------

%-------------------------------------------------------------------
\subsubsection{Outline of the Proposed Approach \label{subsec:outline}}
%-------------------------------------------------------------------

As discussed, the CR video multicast problem is highly challenging since a lot of design choices are tightly coupled.  First, as users see different channels, such heterogeneity should be accommodated so that a user can receive a video quality commensurate to its channel quality. Second, we need to determine the video rates before transmission, which, however, depend on future channel evolution and choice of MC schemes.  Third, the trade-off between primary user protection and spectrum utilization should guide the scheduling of video packets to channels.  Finally, all the optimization decisions should be made in real-time.  Low-complexity, but efficient algorithms are needed, while theoretical optimality bounds would be highly appealing. 

To address heterogeneous user channels, we adopt FGS to produce a base layer with rate $R^b_g$ and an enhancement layer with rate $\bar{R}_g^e$. Without loss of generality, we assume $R^b_g$ is prescribed for an acceptable video quality, while $\bar{R}^e_g$ is set to a large value that is allowed by the codec.  During transmission, we determine the {\em effective rate} for each enhancement layer $R_g^e \leq \bar{R}^e_g$ depending on channel availability, sensing, and MC schemes.\footnote{The proposed approach can also be used for streaming stored FGS video.}  The optimal partition of the enhancement layer should be determined such that each sub-layer uses a different MC scheme.  

We determine the optimal partition of enhancement layers, the choices of MC schemes, and video packet scheduling as follows.  First, we solve the optimal partition problem for every GoP based on an estimated (i.e., average) number of available tiles $T_e$ in the next GoP window that can be used for the enhancement layer, using algorithm GRD1 with complexity $O(MGT_e)$. The tile allocations are then dynamically adjusted in each time slot according to more recent (and thus more accurate) channel status using algorithm GRD2, with complexity $O(MGK)$, where $K \ll T_e$.  Second, during each time slot, video packets are scheduled to the available channels such that the overall system utility is maximized. The TSA algorithm has complexity $O(N\log N)$. Both GRD2 and TSA have low complexity and are suitable for execution in each time slot. 

In real-time video, overdue packets generally do not contribute to improving the received quality.  We assume that the data from a GoP should be be delivered in the next GoP window consisting of $T_{GoP}$ time slots.\footnote{The proposed approach also works for the more general delay requirements that are multiple GoP windows.} Since the base layer is essential for decoding a video, we assume that the base layers of all the videos are coded using $MC_1$. For the $M$ sub-layers of the enhancement layer, a more important sub-layer will be coded using a more reliable (i.e., lower rate) MC scheme.  At the beginning of each GoP window, all the base layers are transmitted using the available tiles.  {\em Retransmissions} will be scheduled if no ACK is received for a base layer packet.  After the base layers are transmitted, we allocate the remaining available tiles in the GoP window for the enhancement layer. The same rule applies to the enhancement sub-layers, such that a higher sub-layer will be transmitted if and only if all the lower sub-layers are acknowledged. This is due to the decoding dependency of layered video. 
  
In each time slot $t$, the base station opportunistically access every channel $n$ with probability $p_n^{tr}(t)$ given in (\ref{eq:ptran}).  Specifically, for each channel $n$, the base station generates a random number $x_n(t)$, which is independent of the channel history $\theta_n(t)$ and uniformly distributed in [0,1].  If $x_n(t) \leq p_n^{tr}(t)$, the most important packet among those not ACKed in the previous GoP will be transmitted on channel $n$. If an ACK is received for this packet at the end of time slot $t$, this packet is successfully received by at least one of the users and will be removed from the transmission buffer. Otherwise, there is a collision with primary user and this packet will remain in the transmission buffer and will be retransmitted. 

In the following, we describe in detail the three algorithms.

%%%%%%%%%%%%%%%%%%%%%%%%%%%%%%%%%%%%%%%%%%%%%%%%%%%%%%%%%%%%%%%%%%%%%%%%%%%%%%%%%%%%
\subsubsection{Enhancement Layer Partitioning and Tile Allocation \label{subsec:part}}

As a first step, we need to determine the effective rate for each enhancement layer $R_g^e \leq \bar{R}^e_g$.  We also need to determine the optimal partition of each enhancement layer. Clearly, the solutions will be highly dependent on the channel availability processes and sensing results.  

Recall that the base layers are transmitted using $MC_1$ first in each GoP window. The {\em remaining} available tiles can then be allocated to the enhancement layers.  We assume that the number of tiles used for the enhancement layers in a GoP window, $T_e$, is known at the beginning of the GoP window.  For example, we can estimate $T_e$ by computing the total average ``idle'' intervals of all the $N$ channels based on the channel model, decreased by the number of tiles used for the base layers (i.e., $R^b_g/b_{g,1}$).  We then split the enhancement layer of each video $g$ into $M$ sub-layers, each occupying $l_{g,m}$ tiles when coded with $MC_m$,  $m=1, 2, \cdots, M$. 

Letting $\vec{l} = [l_{1,1}, l_{1,2}, \cdots, l_{1,M}, l_{2,1}, \cdots l_{G,M}]$ denote the {\em tile allocation vector}, we formulate an optimization problem OPT-Part as follows. 
\begin{eqnarray}
 \mbox{maximize:} && \hspace*{-0.2in} U(\vec{l})
  = \sum_{g=1}^{G}\sum_{k=1}^M(n_{g,k}-n_{g,k+1}) \times 
  %\nonumber \\
  %& & \;\;\;\;\;\;\;\;\;\;\;\;\;\;\; 
  \log \left[ Q^b_g+\beta_g \sum_{m=1}^k b_{g,m}l_{g,m} \right] \label{eq:allocat} \\
\mbox{subject to:} 
  && \hspace*{-0.2in} \sum_{g=1}^G\sum_{m=1}^M l_{g,m}\le T_e \label{eq:st1} \\
	&& \hspace*{-0.2in} \sum_{m=1}^M b_{g,m}l_{g,m} \le \bar{R}_g^e, \; g \in [1, \cdots, M] 
	       \label{eq:st2} \\
	&& \hspace*{-0.2in} l_{g,m} \geq 0, \; m \in [1, \cdots, M], g \in [1, \cdots, G]. 
	       \label{eq:st3}
\end{eqnarray}
OPT-Part is solved at the beginning of each GoP window to determine the optimal partition of the enhancement layer.  The objective is to maximize the overall system utility by choosing optimal values for the $l_{g,m}$'s.  We can derive the effective video rates as $R^e_g = \sum_{m=1}^M b_{g,m} l_{g,m}$.  The formulated problem is a MINLP problem, which is NP-hard~\cite{Deb08}.  In the following, we present two algorithms for computing near-optimal solutions to problem OPT-Part: (i) a {\em sequential fixing} (SF) algorithm based on a linear relaxation of (\ref{eq:allocat}), and (ii) a {\em greedy algorithm} {GRD1} with proven optimality gap.

%%%%%%%%%%%%%%%%%%%%%%%%%%%%%%%%%%%%%%%%%%%%%%%%%%%%%%%%%%%%%%%%%%%%%%%%%%%%%%%%%%%%
\paragraph{A Sequential Fixing Algorithm \label{subsubsec:sf}}

With this algorithm, the original MINLP is first linearized to obtain a linear programming (LP) relaxation. Then we iteratively solve the LP, while fixing one integer variable in every iteration~\cite{Hou08, Hou06}. 
%~\cite{Shi07}.  
We use the {\em Reformulation-Linearization Technique} (RLT) to obtain the LP relaxation~\cite{Kompella09}.  RLT is a technique that can be used to produce LP relaxations for a nonlinear, nonconvex polynomial programming problem. This relaxation will provide a tight upper bound for a maximization problem. Specifically, we linearize the logarithm function in (\ref{eq:allocat}) over some suitable, tightly-bounded interval using a polyhedral outer approximation comprised of a convex envelope in concert with several tangential supports. We further relax the integer constraints, i.e., allowing the $l_{g,m}$'s to take fractional values. Then we obtain an upper-bounding LP relaxation that can be solved in polynomial time. Due to lack of space, we refer interested readers to~\cite{Kompella09} for a detailed description of the technique. 

We next solve the LP relaxation iteratively.  During each iteration, we find the $l_{\hat{g},\hat{m}}$ which has the minimum value for $\left( \lceil l_{\hat{g},\hat{m}} \rceil - l_{\hat{g},\hat{m}} \right)$ or $\left( l_{\hat{g},\hat{m}}-\lfloor l_{\hat{g},\hat{m}} \rfloor \right)$ among all fractional $l_{g,m}$'s, and round it up or down to the nearest integer. We next reformulate and solve a new LP with $l_{\hat{g},\hat{m}}$ fixed. This procedure repeats until all the $l_{g,m}$'s are fixed. The complete SF algorithm is given in Table~\ref{tab:seqfixing0}.  The complexity of SF depends on the specific LP algorithm (e.g., the {\em simplex method} with polynomial-time average-case complexity). 

\begin{table} 
\begin{center}
\caption{The Sequential Fixing (SF) Algorithm}
\begin{tabular}{ll}
\hline
1: & Use RLT to linearize the original problem \\
2: & Solved the LP relaxation \\ %(e.g., using the {\em simplex method}). \\
3: & Suppose $l_{\hat{g},\hat{m}}$ is the integer variable with the minimum \\
   &  \ \ $\left( \lceil l_{\hat{g},\hat{m}} \rceil - l_{\hat{g},\hat{m}} \right)$ 
      or $\left( l_{\hat{g},\hat{m}}-\lfloor l_{\hat{g},\hat{m}} \rfloor \right)$  
      value among all $l_{g,m}$ \\
   &  \ \ variables that remain to be fixed, round it up or down to the \\
   &  \ \ nearest integer \\
4: & If all $l_{g,m}$'s are fixed, got to Step 6 \\
5: & Otherwise, reformulate a new relaxed LP with the newly \\
   & \ \ fixed $l_{g,m}$ variables, and go to Step 2 \\
6: & Output all fixed $l_{g,m}$ variables and $R^e_g = \sum_{m=1}^M b_{g,m} l_{g,m}$ \\
\hline
\end{tabular}
\label{tab:seqfixing0}
\end{center}
\end{table}

%%%%%%%%%%%%%%%%%%%%%%%%%%%%%%%%%%%%%%%%%%%%%%%%%%%%%%%%%%%%%%%%%%%%%%%%%%%%%%%%%%%%
\paragraph{A Greedy Algorithm \label{subsubsec:ga1}}

Although SF can compute a near-optimal solution in polynomial time, it does not provide any guarantee on the optimality of the solution.  In the following, we describe a greedy algorithm, termed GRD1, which exploits the inherent priority structure of layered video and MC schemes and has a proven optimality bound.  

The complete greedy algorithm is given in Table~\ref{tab:GreedyAlgo1}, where $R=\sum_{g=1}^G \bar{R}_g^e$ is the total rate of all the enhancement layers and $\vec{e}_{i}$ is a {\em unit vector} with ``1'' at the $i$-th location and ``0'' at all other locations.  In GRD1, all the $l_{g,m}$'s are initially set to 0. During each iteration, one tile is allocated to the $\hat{m}$-th sub-layer of video $\hat{g}$.  In Step 4, $l_{\hat{m},  \hat{g}}$ is chosen to be the one that achieves the largest increase in terms of the ``normalized'' utility (i.e., $[U(\vec{l}+\vec{e}_{g,m})-U(\vec{l})]/[b_{g,m}+R/T_e]$) if it is assigned with an additional tile.  Lines 6, 7, and 8 check if the assigned rate exceeds the maximum rate $\bar{R}^e_g$.  GRD1 terminates when either all the available tiles are used or when all the video data are allocated with tiles. In the latter case, all the videos are transmitted at full rates.  We have the following Theorem for GRD1.  
\begin{table} [!t]
\begin{center}
\caption{The Greedy Algorithm (GRD1)}
\begin{tabular}{ll}
\hline
1:  & Initialize $l_{g,m}=0$ for all $g$ and $m$ \\
2:  & Initialize $A=\{1,2,\cdots,G\}$\\
3:  & WHILE $\left( \sum_{g=1}^G\sum_{m=1}^Ml_{g,m} \le T_e \; \mbox{and } A \; 
         \mbox{is not empty} \right)$ \\
4:  & $\;\;$ Find $l_{\hat{g},\hat{m}}$ that can be increased by one: \\
    & $\;\;\;\;\;\;$ $\vec{e}_{\hat{g},\hat{m}}=\arg \max_{g \in A,m \in [1,\cdots,M]} \left\{ \frac{U(\vec{l}+\vec{e}_{g,m})-U(\vec{l})}{b_{g,m}+R/T_e} \right\}$ \\
5:  & $\;\;$ $\vec{l}=\vec{l}+\vec{e}_{\hat{g},\hat{m}}$  \\
6:  & $\;\;$ IF $\left( \sum_m b_{\hat{g},m}l_{\hat{g},m} > \bar{R}_g^e \right)$ \\
7:  & $\;\;\;\;\;\;$ $\vec{l}=\vec{l}-\vec{e}_{\hat{g},\hat{m}}$  \\
8:  & $\;\;\;\;\;\;$ Delete $\hat{g}$ from $A$ \\
9:  & $\;\;$ END IF \\
10: & END WHILE \\
\hline
\end{tabular}
\label{tab:GreedyAlgo1}
\end{center}
\end{table}

\smallskip
\begin{theorem} 
The greedy algorithm GRD1 shown in Table~\ref{tab:GreedyAlgo1} has a complexity $O(M G T_e)$.  It guarantees a solution that is within a factor of $(1-e^{-1/2})$ of the global optimal solution. 
\end{theorem}

\begin{proof}
(i) {\em Complexity}: In Step 4 in Table~\ref{tab:GreedyAlgo1}, it takes $O(MG)$ to solve for $\vec{e}_{\hat{g},\hat{m}}$. Since each iteration assigns one tile to sub-layer $\hat{m}$ of group $\hat{g}$, it takes $T_e$ iterations to allocate all the available tiles in a GoP window. Therefore, the overall complexity of GRD1 is $O(MGT_e)$.

(ii) {\em Optimality Bound}: This proof is extended from a result first shown in~\cite{Deb08} for layered videos. We first show a property of group utility $U_g(\vec{l})$, which will be used in the proof of the optimality gap.  For two vectors $\vec{l}^1_g$ and $\vec{l}^2_g$, letting $\Delta = U_g(\vec{l}^1_g)-U_g(\vec{l}^2_g)$, we have
%\begin{eqnarray}\label{eq:lema2}
%& & \hspace{-0.2in} \Delta = \mbox{$\sum_{k=1}^M$} (n_{g,k}-n_{g,k+1}) \times \nonumber \\
%& & \hspace{0.2in} \log \left( 1 + \frac{ \sum_{m=1}^k \beta_g b_{g,m}(l_{g,m}^1-l_{g,m}^2) }{ Q^b_g + \sum_{m=1}^k \beta_g b_{g,m} l_{g,m}^2 } \right) \nonumber  \\
%& & \hspace{-0.2in} \le \mbox{$\sum_{k=1}^M \sum_{m=1}^k$} (l_{g,m}^1-l_{g,m}^2)^+(n_{g,k}-n_{g,k+1}) \times \nonumber \\
%& & \hspace{0.2in} \log \left( 1+ \beta_g b_{g,m} / \left[ Q^b_g+\mbox{$\sum_{m=1}^k$} \beta_g b_{g,m} l_{g,m}^2 \right] \right) \nonumber \\
%& & \hspace{-0.2in} \le \mbox{$\sum_{k=1}^M \sum_{m=1}^M$} (l_{g,m}^1-l_{g,m}^2)^+(n_{g,k}-n_{g,k+1}) \times \nonumber \\
%& & \hspace{0.2in} \log \left( 1 + \beta_g b_{g,m} / \left[ Q^b_g + \mbox{$\sum_{m=1}^k$} \beta_g b_{g,m} l_{g,m}^2 \right] \right) \nonumber\\
%& & \hspace{-0.2in} = \mbox{$\sum_{m=1}^M$} (l_{g,m}^1-l_{g,m}^2)^+ \left[ U_g(\vec{l}^2_g+b_{g,m})-U(\vec{l}^2_g) \right],
%\end{eqnarray}
\begin{eqnarray}\label{eq:lema2}
& & \hspace{-0.2in} \Delta = \sum_{k=1}^M (n_{g,k}-n_{g,k+1}) \times \log \left( 1 + \frac{ \sum_{m=1}^k \beta_g b_{g,m}(l_{g,m}^1-l_{g,m}^2) }{ Q^b_g + \sum_{m=1}^k \beta_g b_{g,m} l_{g,m}^2 } \right) \nonumber  \\
& & \hspace{-0.2in} \le \sum_{k=1}^M \sum_{m=1}^k (l_{g,m}^1-l_{g,m}^2)^+(n_{g,k}-n_{g,k+1}) \times \log \left( 1+ \beta_g b_{g,m} / \left[ Q^b_g+\sum_{m=1}^k \beta_g b_{g,m} l_{g,m}^2 \right] \right) \nonumber \\
& & \hspace{-0.2in} \le \sum_{k=1}^M \sum_{m=1}^M (l_{g,m}^1-l_{g,m}^2)^+(n_{g,k}-n_{g,k+1}) \times  \log \left( 1 + \beta_g b_{g,m} / \left[ Q^b_g + \sum_{m=1}^k \beta_g b_{g,m} l_{g,m}^2 \right] \right) \nonumber\\
& & \hspace{-0.2in} = \sum_{m=1}^M (l_{g,m}^1-l_{g,m}^2)^+ \left[ U_g(\vec{l}^2_g+b_{g,m})-U(\vec{l}^2_g) \right],
\end{eqnarray}
where $y^+=\max \{ 0, y \}$. The first inequality is due to the concavity of logarithm functions. 

Next we prove the optimality bound. Let $\vec{l}_{t}$ be the output of GRD1 after $t$ iterations. Let the utility gap between the optimal solution and the GRD1 solution be $F_{t} = U(\vec{l}^{\ast}) - U(\vec{l}_{t})$, and $\vec{e}_{\hat{g},\hat{m}}(t)$ the argument found in Step 4 of GRD1 after $t$ iterations. We have $\vec{l}_{t} = \vec{l}_{t-1} + \vec{e}_{\hat{g}, \hat{m}}(t)$ and
\begin{eqnarray}
 & & \hspace{-0.1in} F_{t-1} = U(\vec{l}^{\ast})-U(\vec{l}_{t-1}) \nonumber \\
 & & \hspace{-0.3in} \le \sum_g\sum_m (l_{g,m}^{\ast}-l_{g,m})^{+}
[U(\vec{l}_{t-1}+\vec{e}_{g,m}(t))-U(\vec{l}_{t-1})] \nonumber \\
 & & \hspace{-0.3in} \le \sum_g\sum_m (l_{g,m}^{\ast}-l_{g,m})^{+} [U(\vec{l}_{t-1}+\vec{e}_{\hat{g},\hat{m}}(t))- U(\vec{l}_{t-1})]
\frac{b_{g,m}+R/T_e}{b_{\hat{g},\hat{m}}(t)+R/T_e}\nonumber \\
 & & \hspace{-0.3in} \le \frac{U(\vec{l}_t)-U(\vec{l}_{t-1})}{ b_{\hat{g}, \hat{m}}(t) + R/T_e}
\sum_g\sum_m [l_{g,m}^{\ast}(b_{g,m}+R/T_e)]. \nonumber
\end{eqnarray}
The first inequality is due to (\ref{eq:lema2}) and the second inequality follows Step 4 of GRD1.  It follows (\ref{eq:st1}) that $\sum_g\sum_m l_{g,m}^{\ast} \le T_e$ and $\sum_g\sum_m b_{g,m}l_{g,m}^{\ast} \le R$.  We have $F_{t-1} \le (F_{t-1}-F_{t}) \frac{2R}{b_{\hat{g}, \hat{m}}(t)+R/T_e}$.  Solving for $F_t$, we have $F_{t}\le F_{t-1} \left\{ 1 - \left[ b_{\hat{g},\hat{m}}(t)+R/T_e \right] / (2R) \right\}$. 

Suppose the {\em WHILE} loop in Table~\ref{tab:GreedyAlgo1} has been executed $k$ times when the solution is obtained. 
\begin{eqnarray}
F_k &\le& F_{k-1} \left\{ 1-\left[b_{\hat{g},\hat{m}}(k)+R/T_e\right] / (2R) \right\} \nonumber \\
& \le & F_0 \prod_{t=1}^k \left\{ 1 - \left[ b_{\hat{g},\hat{m}}(t)+R/T_e \right] / (2R) \right\} \nonumber \\
&\le& F_0 \left\{ 1 - 1/(2kR) \sum_{t=1}^k [ b_{\hat{g},\hat{m}}(t)+R/T_e ] \right\}^k. \nonumber
\end{eqnarray}
The {\em WHILE} loop exits when one or both of two constraints are violated. If $\sum_{g}\sum_{m} l_{g,m} \le T_e$ is violated, there is no tile that can be used. Therefore $k \ge T_e$ and $\sum_{t=1}^kR/T_e\ge R$. If constraint ``$A$ is not empty'' is violated, all the videos have been allocated sufficient number of tiles and will be transmitted at full rates. We have $\sum_{t=1}^k b_{\hat{g},\hat{m}}(t) \geq R$ in this case.  It follows that
\begin{eqnarray}
F_k &\le& 
F_0 \left\{ 1 - 1/(2kR) \sum_{t=1}^k [ b_{\hat{g},\hat{m}}(t)+R/T_e ] \right\}^k \nonumber \\
& \le & F_0 \left[ 1 - 1/(2k) \right]^k \; \le \; F_0 e^{-1/2}. \nonumber %\\
\end{eqnarray}
Since $F_0=U(\vec{l}^{\ast})$, we have $U(\vec{l}_k)\ge(1-e^{-1/2})U(\vec{l}^{\ast})$.  Therefore, we conclude that the GRD1 solution is bounded by $(1-e^{-1/2})U(\vec{l}^{\ast})$ and $U(\vec{l}^{\ast})$.
\end{proof}

%%%%%%%%%%%%%%%%%%%%%%%%%%%%%%%%%%%%%%%%%%%%%%%%%%%%%%%%%%%%%%%%%%%%%%%%%%%%%%%%%%%%
\paragraph{A Refined Greedy Algorithm \label{subsubsec:ga2}}

GRD1 computes $l_{g,m}$'s based on an estimate of network status $\vec{S}(t)$ in the next $T_{GoP}$ time slots.  Due to channel dynamics, the computed $l_{g,m}$'s may not be exactly accurate, especially when $T_{GoP}$ is large.  We next present a refined greedy algorithm, termed GRD2, which adjusts the $l_{g,m}$'s based on more accurate estimation of the channel status.

GRD2 is executed at the beginning of every time slot.  It estimates the number of available tiles $T_e(t)$ in the next $T_{est}$ successive time slots, where $1 \le T_{est}\le T_{GoP}$ is a design parameter depending on the coherence time of the channels. Such estimates are more accurate than that in GRD1 since they are based on recently received ACKs and recent sensing results.  Specifically, we estimate $T_e(t)$ using the belief vector $\vec{a}(t)$ in time slot $t$. Recall that $a_n(t)$'s are computed based on the channel model, feedback, sensing results, and sensing errors, as given in (\ref{eq:AvailProb1}), and (\ref{eq:OneStep}). For the next time slot, $a_n(t+1)$ can be estimated as $\hat{a}_n(t+1) = \lambda_n a_n(t) + \mu_n [1-a_n(t)] = (\lambda_n-\mu_n) a_n(t)+\mu_n$.  Recursively, we can derive $\hat{a}_n(t+\tau)$ for the next $\tau$ time slots.
\begin{eqnarray}
  \hat{a}_n(t+\tau) = (\lambda_n-\mu_n)^\tau a_n(t)+ %\nonumber \\
  \mu_n \frac{1-(\lambda_n-\mu_n)^\tau}{1-(\lambda_n-\mu_n)}.
\end{eqnarray}

At the beginning section of a GoP window, all the base layers will be firstly transmitted. We start the estimation after all the base layers have been successfully received (possibly with retransmissions). The number of available tiles in the following $T_{est}$ time slots can be 
estimated as $T_e(t) = \sum_{n=1}^N\sum_{\tau=0}^{t_{min}}\hat{a}_n(t+\tau)$, where $\hat{a}_n(t+0)=a_n(t)$ and $t_{min}=\min\{T_{est}-1,T_{GoP} - (t \; \mbox{mod} \; T_{GoP})\}$.  $T_e(t)$ may not be an integer, but it does not affect the outcome of GRD2.  

We then adjust the $l_{g,m}$'s based on $T_e(t)$ and $N_{ack}(t)$, the number of ACKs received in time slot $t$.  If $T_e(t)+N_{ack}(t-1) > T_e(t-1)+N_{ack}(t-2)$, there are more tiles that can be allocated and we can increase some of the $l_{g,m}$'s. On the other hand, if $T_e(t)+N_{ack}(t-1) < T_e(t-1)+N_{ack}(t-2)$, we have to reduce some of the $l_{g,m}$'s.  Due to layered videos, when we increase the number of allocated tiles, we only need to consider $l_{g,m}$ for $m=m', m'+1, \cdots, M$, where $MC_{m'}$ is the highest MC scheme used in the previous time slot.  Similarly, when we reduce the number of allocated tiles, we only need to consider $l_{g,m}$ for $m=m',m'+1,\cdots,M$. 

The refined greedy algorithm is given in Table~\ref{tab:GreedyAlgo2}. For time slot $t$, the complexity of GRD2 is $O(MGK)$, where $K = | N_{ack}(t-1) - N_{ack}(t-2) + T_e(t)-T_e(t-1) |$.  Since $K \ll T_e$, the complexity of GRD2 is much lower than GRD1, suitable for execution in each time slot.

\begin{table} [!t]
\begin{center}
\caption{The Refined Greedy Algorithm (GRD2) for Each Time Slot}
\begin{tabular}{ll}
\hline
1: & Initialize $l_{g,m}=0$ for all $g$ and $m$ \\
2: & Initialize $A=\{1,2,\cdots,G\}$\\
3: & Initialize $N_{ack}(0)=0$ \\
4: & Estimate $T_e(1)$ based on the Markov Chain channel model \\
5: & Use GRD1 to find all $l_{g,m}$'s based on $T_e(1)$ \\
6: & WHILE $t=2$ to $T_{GoP}$ \\
7: & \ \ Estimate $T_e(t)$ \\
8: & \ \ IF $\left[ T_e(t)+N_{ack}(t-1) < T_e(t-1)+N_{ack}(t-2) \right]$ \\
9: & \ \ \ \ WHILE $\left[ \sum_{g=1}^G\sum_{m=1}^Ml_{g,m} > T_e(t)+N_{ack}(t-2) \right]$ \\
10: & \ \ \ \ \ \ Find $l_{\hat{g},\hat{m}}$ that can be reduced by 1: \\
  & \ \ \ \ \ \ \ \ $\vec{e}_{\hat{g},\hat{m}}=\arg \min_{\forall g, m \in \{m',\cdots,M\}} \left\{\frac{U(\vec{l})-U(\vec{l}-\vec{e}_{g,m})}{b_{g,m}+R/T_e} \right\}$ \\
11: & \ \ \ \ \ \ $\vec{l}=\vec{l}-\vec{e}_{\hat{g},\hat{m}}$ \\
12: & \ \ \ \ \ \ IF $(\hat{g} \notin A)$\\
13: & \ \ \ \ \ \ \ \ Add $\hat{g}$ to $A$\\
14: & \ \ \ \ \ \ END IF \\
15: & \ \ \ \ END WHILE \\
16: & \ \ END IF \\
17: & \ \ IF $\left[ T_e(t)+N_{ack}(t-1) > T_e(t-1)+N_{ack}(t-2) \right]$ \\
18: & \ \ \ \ WHILE $\left[ \sum_{g=1}^G\sum_{m=1}^Ml_{g,m} \le T_e(t)+N_{ack}(t-1) \right.$ and \\
    & \ \ \ \ \ \  $\left. A \mbox{ is not empty} \right]$ \\
19: & \ \ \ \ \ \ Find $l_{\hat{g},\hat{m}}$ that can be increased by 1 \\
    & \ \ \ \ \ \ \ \ $\vec{e}_{\hat{g},\hat{m}}=\arg \max_{g \in A,m \in \{m',\cdots,M\}} \left\{\frac{U(\vec{l}+\vec{e}_{g,m})-U(\vec{l})}{b_{g,m}+R/T_e} \right\}$ \\
20: & \ \ \ \ \ \ $\vec{l}=\vec{l}+\vec{e}_{\hat{g},\hat{m}}$                             \\
21: & \ \ \ \ \ \ IF $\left( \sum_m b_{\hat{g},m}l_{\hat{g},m} > \bar{R}^e_g \right)$ \\
22: & \ \ \ \ \ \ \ \ $\vec{l}=\vec{l}-\vec{e}_{\hat{g},\hat{m}}$ \\
23: & \ \ \ \ \ \ \ \ Delete $\hat{g}$ from $A$ \\
24: & \ \ \ \ \ \ END IF \\
25: & \ \ \ \ END WHILE \\
26: & \ \ END IF \\
27: & \ \ Update $N_{ack}(t-1)$ \\
28: & END WHILE \\
\hline
\end{tabular}
\label{tab:GreedyAlgo2}
\end{center}
\end{table}

%%%%%%%%%%%%%%%%%%%%%%%%%%%%%%%%%%%%%%%%%%%%%%%%%%%%%%%%%%%%%%%%%%%%%%%%%%%%%%%%%%%%
\subsubsection{Tile Scheduling in a Time Slot \label{subsec:schedule}}

In each time slot $t$, we need to schedule the remaining tiles for transmission on the $N$ channels.  We define $\mbox{Inc}(g,m,i)$ to be the increase in the group utility function $U(g)$ after the $i$-th tile in the sub-layer using $MC_m$ is successfully decoded. It can be shown that
%\begin{eqnarray}
%& & \hspace*{-0.2in} 
$$
\mbox{Inc}(g,m,i) = \sum_{k=m}^M (n_{g,k}-n_{g,k+1}) \times 
%\nonumber \\
%& & 
\log \left[ 1+\frac{\beta_gb_{g,m}}{Q_g^b+\beta_g\sum_{u=1}^{m-1}b_{g,u}l_{g,u}+(i-1)\beta_gb_{g,m}} \right]. \nonumber
$$
%\end{eqnarray}
$\mbox{Inc}(g,m,i)$ can be interpreted as the {\em reward} if the tile is successfully received. 

Letting $c_n(t)$ be the probability that the tile is successfully received, then we have $c_n(t) = p^{tr}_n(t) a_n(t)$.  Our objective of tile scheduling is to maximize the expected reward, i.e.,
\begin{eqnarray} \label{eq:mreward}
  \mbox{maximize:} \;\;\; \mbox{E} [\mbox{Reward}(\vec{\xi})] = \sum_{n=1}^N c_n(t) \cdot \mbox{Inc}(\xi_n),
\end{eqnarray}
where $\vec{\xi} = \{ \xi_n \}_{n=1,\cdots,N}$ and $\xi_n$ is the tile allocation for channel $n$, i.e., representing the three-tuple $\{g,m,i\}$.  The TSA algorithm is shown in Table~\ref{tab:GreedyAlgo3}, which solves the above optimization problem.  The complexity of TSA is $O(N \log N)$. We have the following theorem for TSA. 

\begin{table} [!t]
\begin{center}
\caption{Algorithm for Tile Scheduling in a Time Slot}
\begin{tabular}{ll}
\hline
1: & Initialize $m_g$ to the lowest MC that has not been ACKed for all $g$ \\ 
2: & Initialize $i_g$ to the first packet that has not been ACKed for all $g$ \\ 
3: & Sort $\{c_n(t)\}$ in decreasing order. Let the sorted channel list be \\
   & \ \ indexed by $j$. \\ 
4: & While ($j=1$ to $N$) \\
5: & \ \ Find the group having the maximum increase in $U(g)$: \\
   & \ \ \ \ \ $\hat{g}=\arg \max_{\forall g} \mbox{Inc}(g,m_g,i_g)$ \\
6: & \ \ Allocate the tile on channel $j$ to group $\hat{g}$ \\
7: & \ \ Update $m_{\hat{g}}$ and $i_{\hat{g}}$ \\
8: & End while \\
\hline
\end{tabular}
\label{tab:GreedyAlgo3}
\end{center}
\end{table}

\begin{theorem}
  E[Reward] is maximized if $\mbox{Inc}(\xi_i) > \mbox{Inc}(\xi_j)$ when $c_i(t) > c_j(t)$ 
  for all $i$ and $j$.
\end{theorem}
\begin{proof} Suppose there exists a pair of $i$ and $j$ where $\mbox{Inc}(\xi_i) > \mbox{Inc}(\xi_j)$ and $c_i(t) < c_j(t)$. We can further increase E[Reward] by switching the tile assignment, i.e., assign channel $i$ to $\xi_j$ and channel $j$ to $\xi_i$. With this new assignment, the net increase in E[Reward] is 
\begin{eqnarray}
  & & \hspace*{-0.3in} c_j(t)\mbox{Inc}(\xi_i) + c_i(t)\mbox{Inc}(\xi_j)-c_i(t)\mbox{Inc}(\xi_i)-c_j(t)\mbox{Inc}(\xi_j) \nonumber \\
  & = & [c_j(t) - c_i(t)][\mbox{Inc}(\xi_i) - \mbox{Inc}(\xi_j)] > 0. \nonumber
\end{eqnarray}
Therefore E[Reward] is maximized when the $\{ \mbox{Inc}(\xi_i) \}$ and $\{c_i(t)\}$ are in the same order. 
\end{proof}

%%%%%%%%%%%%%%%%%%%%%%%%%%%%%%%%%%%%%%%%%%%%%%%%%%%%%%%%%%%%%%%%%%%%%%%%%%%%%%%%%%%%
\subsection{Simulation Results \label{sec:sim}}

We evaluate the performance of the proposed CR video multicast framework using a customized simulator implemented with a combination of C and MATLAB. Specifically, the LPs are solved using the MATLAB Optimization Toolbox and the remaining parts are written in C. For the results reported in this section, we have $N=12$ channels (unless otherwise specified). The channel parameters $\lambda_n$ and $\mu_n$ are set between $(0, 1)$. The maximum allowed collision probability $\gamma_n$ is set to 0.2 for all the channels unless  otherwise specified. 

The CR base station multicasts three Common Intermediate Format (CIF, $352 \times 288$) video sequences to three multicast groups, i.e., {\em Bus} to group 1, {\em Foreman} to group 2, and {\em Mother \& Daughter} to group 3.  The $n_{1,m}$'s are \{42, 40, 36, 30, 22, 12\} (i.e., 42 users can decode $MC_1$ signal, 40 users can decode $MC_2$ signal, and so forth); the $n_{2,m}$'s are \{51, 46, 40, 32, 23, 12\} and the $n_{3,m}$'s are \{49, 44, 40, 32, 24, 13\}.  The number of bits carried in one tile using the MC schemes are 1 kb/s, 1.5 kb/s, 2 kb/s, 3 kb/s, 5.3 kb/s, and 6 kb/s, respectively.  We choose $T_{GoP}$=150 and $T_{est}=10$, sensing interval $W=3$, false alarm probability $\epsilon_n = 0.3$ and miss detection probability $\delta_n = 0.25$ for all $n$, unless otherwise specified. 

In every simulation, we compare three schemes: (i) a simple heuristic scheme that equally allocates tiles to each group (Equal Allocation); (ii) A scheme based on SF (Sequential Fixing), and (iii) a scheme based on the greedy algorithm GRD2 (Greedy Algorithm).  These schemes have increasing complexity in the order of Equal Allocation, Greedy Algorithm, and Sequential Fixing. They differ on how to solve Problem OPT-Part, while the same tile scheduling algorithm and opportunistic spectrum access scheme are used in all the schemes.  Each point in the figures is the average of 10 simulation runs, with 95\% confidence intervals plotted.  We observe that the 95\% confidence intervals for Equal Allocation and Greedy Algorithm are negligible, while the 95\% confidence intervals for Sequential Fixing is relatively larger.  The C/MATLAB code is executed in a Dell Precision Workstation 390 with an Intel Core 2 Duo E6300 CPU working at 1.86 GHz and a 1066 MB memory.  For number of channels ranging from $N$=3 to $N$=15, the execution times of Equal Allocation and Greedy Algorithm are about a few milliseconds, while Sequential Fixing takes about two seconds.  

In Fig.~\ref{fig:PSNRcomp} we plot the average PSNR among all users in each multicast group. 
For all the groups, Greedy Algorithm achieves the best performance, with up to 4.2 dB improvements over Equal Allocation and up to 0.6 dB improvements over Sequential Fixing.  We find Sequential Fixing achieves a lower PSNR than Equal Allocation for group 3, but higher PSNRs for groups 1 and 2. This is because Equal Allocation does not consider channel conditions and fairness.  It achieves better performance for group 3 at the cost of much lower PSNRs for groups 1 and 2.  We also plot Frame 53 from the original {\em Bus} sequence and the decoded video at user 1 of group 1 in Fig~\ref{fig:images}. We choose this user since it is one of the users with lowest PSNR values. The average PSNR of this user is 29.54 dB, while the average PSNR of all group 1 users is 34.6 dB.  Compared to the original frame (right), the reconstructed frame (left) looks quite good, although some details are lost.

%--------------------------------------------------------------------------------
\begin{figure}[!t]
\centering
\includegraphics[width=4.5in, height=3.0in]{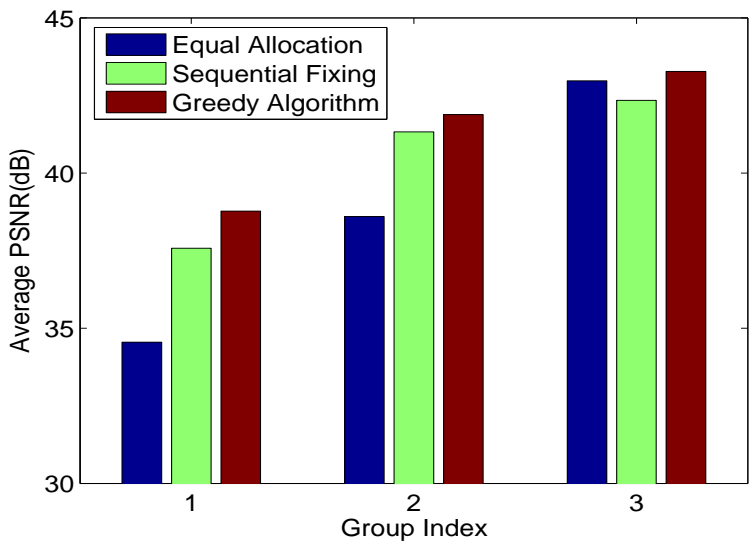}
\caption{Average PSNR of all multicast users.}
\label{fig:PSNRcomp}
\end{figure}
%--------------------------------------------------------------------------------

%--------------------------------------------------------------------------------
\begin{figure}[!t]
\centering
\includegraphics[width=4.5in]{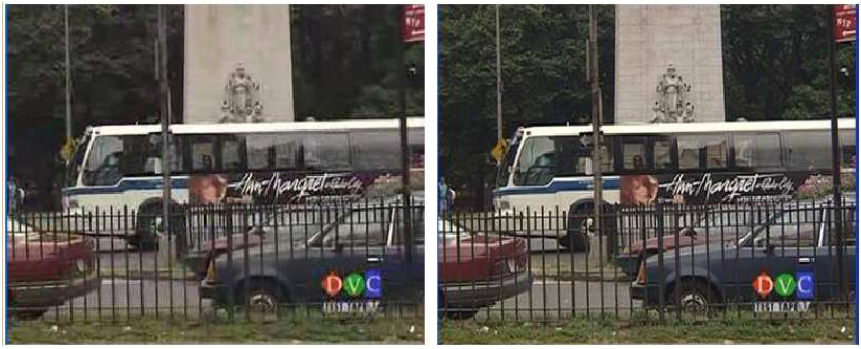}
\caption{The original (the right one) and decoded Frame 53 (the left one) at user 1 in group 1.}
\label{fig:images}
\end{figure}
%--------------------------------------------------------------------------------

In Fig.~\ref{fig:PSNRgamma}, we examine the impact of the maximum allowed collision probability $\gamma_n$.  We increase $\gamma_n$ from 0.1 to 0.3, and plot the average PSNR values among all the users.  When $\gamma_n$ gets larger, there will be higher chance of collision for the video packets, which hurts the received video quality.  However, a higher $\gamma_n$ also allows a higher transmission probability $p_n^{tr}(t)$ for the base station (see (\ref{eq:ptran})), thus allowing the base station to grab more spectrum opportunities and achieve a higher video rate.  The net effect of these two contradicting effects is improved video quality for the range of $\gamma_n$ values considered in this simulation.  This is illustrated in the figure where all the three curves increase as $\gamma_n$ gets larger.  We also observe that the curves for Sequential Fixing and Equal Allocation are roughly parallel to each other, while the Greedy Algorithm curve has a steeper slope. This indicates that Greedy Algorithm is more efficient in exploiting the additional bandwidth allowed by an increased $\gamma_n$. 

%--------------------------------------------------------------------------------
\begin{figure}[!t]
\centering
\includegraphics[width=4.5in, height=3.0in]{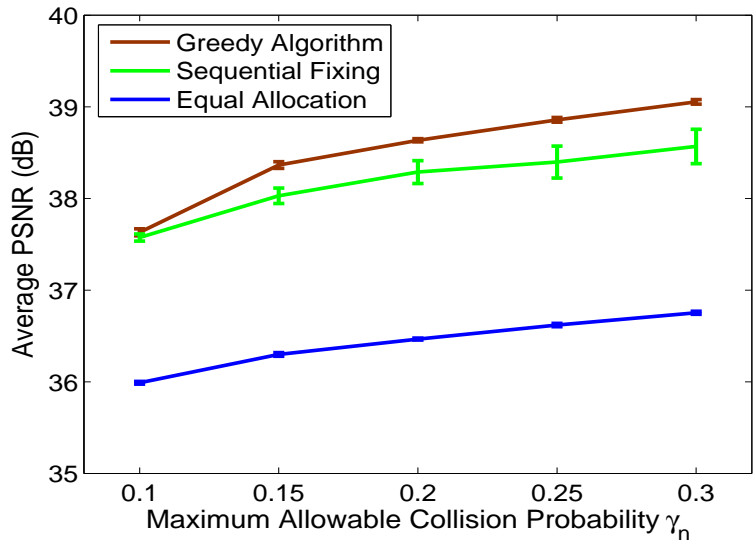}
\caption{Average PSNR of all users versus $\gamma_n$ (with 95\% confidence intervals).}
\label{fig:PSNRgamma}
\end{figure}
%--------------------------------------------------------------------------------

In Fig.~\ref{fig:PSNRn}, we examine the impact of number of channels $N$.  We increase $N$ from 3 to 15 in steps of 3, and plot the average PSNR values of all multicast users. As expected, the more channels, the more spectrum opportunities for the CR networks, and the better the video quality.  Again, we observe that the Greedy Algorithm curve has the steepest slope, implying it is more efficient in exploiting the increased spectrum opportunity for video transmissions. 

%--------------------------------------------------------------------------------
\begin{figure}[!t]
\centering
\includegraphics[width=4.5in, height=3.0in]{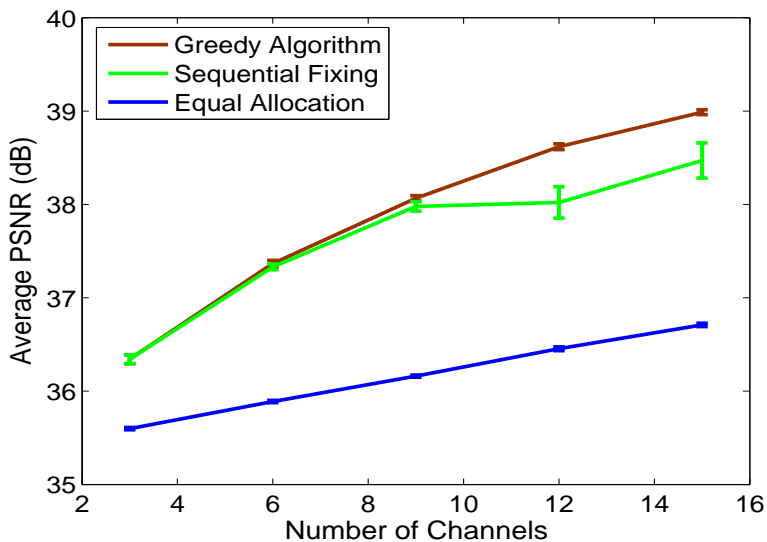}
\caption{Average PSNR of all users versus $N$ (with 95\% confidence intervals).}
\label{fig:PSNRn}
\end{figure}
%--------------------------------------------------------------------------------

We demonstrate the impact of sensing errors in Fig.~\ref{fig:PSNRepsilon}. We test five sets of $\{ \epsilon_n, \delta_n \}$ values as follows: $\{0.10, 0.38\}$, $\{0.30, 0.25\}$, $\{0.5, 0.17\}$, $\{0.70, 0.10\}$  and $\{0.9, 0.04\}$~\cite{Chen08}, and plot the average PSNR values of all users. It is quite interesting to see that the video quality is not very sensitive to sensing errors. Even as $\epsilon_n$ is increased nine times from 10\% to 90\%, there is only 0.58 dB reduction (or a 1.5\% normalized reduction) in average PSNR when Greedy Algorithm is used. The same can be observed for the other two curves.  We conjecture that this is due to the opportunistic spectrum access approach adopted in all the three schemes.  A special strength of the proposed approach is that it explicitly considers both types of sensing errors and mitigates the impact of both sensing errors.  For example, when the false alarm rate is very high, the base station will not trust the sensing results and will access the channel relatively more aggressively, thus mitigating the negative effect of the high false alarm rate.

%--------------------------------------------------------------------------------
\begin{figure}[!t]
\centering
\includegraphics[width=4.5in, height=3.0in]{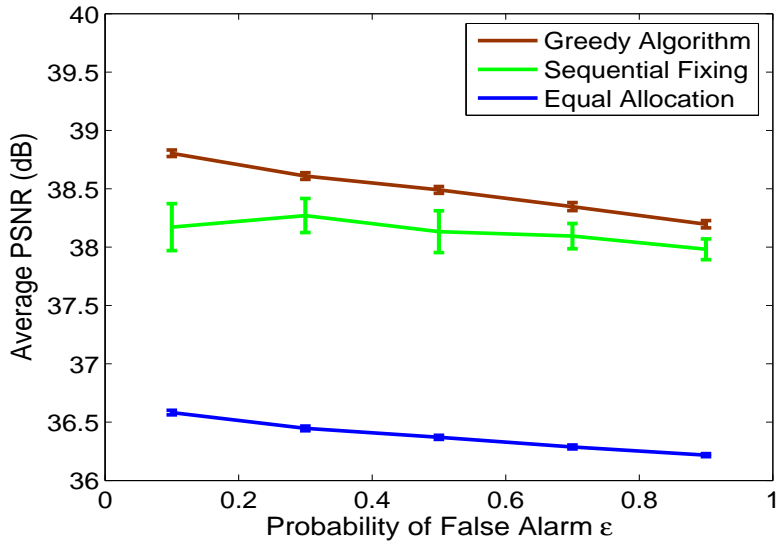}
\caption{Average PSNR of all users for various $\{\epsilon_n, \delta_n\}$ values (with 95\% confidence intervals).}
\label{fig:PSNRepsilon}
\end{figure}
%--------------------------------------------------------------------------------

Finally, we demonstrate the impact of user channel variations (i.e., due to mobility).  We chose a tagged user in group 1 and assume that its channel condition changes every 20 GoPs. The highest MC scheme that the tagged user can decode is changed according to the following sequence: MC3, MC5, MC4, MC6, MC5 and MC3. All other parameters remain the same as in the previous experiments. In Fig.~\ref{fig:PSNRtrace}, we plot the average PSNRs for each GoP at this user that are obtained using the three algorithms. We observe that both Greedy Algorithm and Sequential Fixing can quickly adapt to changing channel conditions. Both algorithms achieve received video qualities commensurate with the channel quality of the tagged user. We also find the video quality achieved by Greedy Algorithm is more stable than that of Sequential Fixing, while the latter curve has some deep fades from time to time. This is due to the fact that Greedy Algorithm has a proven optimality bound, while Sequential Fixing does not provide any guarantee. The Equal Allocation curve is relative constant for the entire period since it does not adapt to channel variations. Although being simple, it does not provide good video quality in this case. 

For optimization-driven multimedia systems, there is a trade-off between (i) grabbing all the available resource to maximize media quality and (ii) be less adaptive to network dynamics for a smooth playout.  The main objective of this paper is to demonstrate the feasibility and layout the framework for video streaming over infrastructure-based CR networks, using an objective function of maximizing the overall user utility. We will investigate the interesting problem of trading off resource utilization and smoothness in our future work.

%--------------------------------------------------------------------------------
\begin{figure}[!t]
\centering
\includegraphics[width=5.7in]{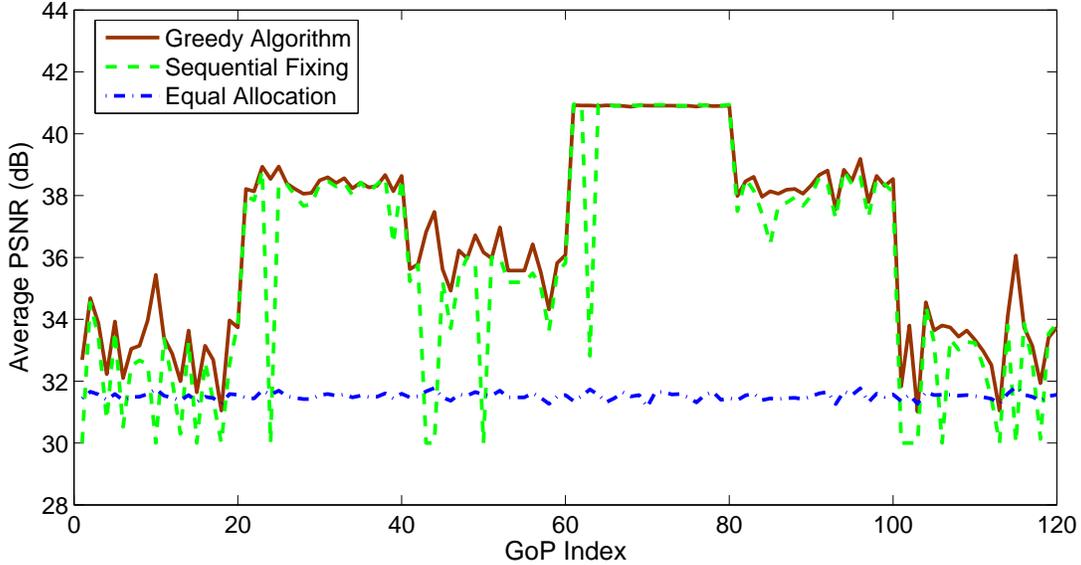}
\caption{GoP average PSNRs of a tagged user in Group 1, when its channel condition varies over time.}
\label{fig:PSNRtrace}
\end{figure}
%--------------------------------------------------------------------------------

%-------------------------------------------------------------------
% Video over multi-hop CR networks
\section{Video over Multi-hop CR Networks}\label{sec:cr_video_mhop}
In this section, we examine the problem of video over multi-hop CR networks. We model streaming of concurrent videos as an MINLP problem, aiming to maximize the overall received video quality and fairness among the video sessions, while bound the collision rate with primary users under spectrum sensing errors. We solve the
MINLP problem using a centralized sequential fixing algorithm, and derive upper and lower bounds for the objective value. We then apply dual decomposition to develop a distributed algorithm and prove its optimality and convergence conditions.

%-------------------------------------------------------------------
\subsection{Network Model}
%-------------------------------------------------------------------

%-------------------------------------------------------------------
\subsubsection{Spectrum Access \label{subsec:dsa}}
%-------------------------------------------------------------------

During the transmission phase of a time slot, a CR user determines which channel(s) to access for transmission of video data based on spectrum sensing results.  Let $\kappa_m^k$ be a threshold  for spectrum access: channel $m$ is considered idle if the estimate $a_m^k$ is greater than the threshold, and busy otherwise. The availability of channel $m$ in primary network $k$, denoted as $A_m^k$, is
\begin{eqnarray}\label{eq:DefAmk}
	A_m^k=\left\{
	\begin{array}{ll}
	0, & \mbox{$a_m^k \ge \kappa_m^k$}\\
	1, & \mbox{otherwise.}\\
\end{array}\right.
\end{eqnarray}

For each channel $m$, we can calculate the probability of collision with primary users as:
\begin{eqnarray}\label{eq:ProbLinkCol}
  \Pr(A_m^k=0|H_1) = \hspace{-0.05in} \sum_{i \in \psi_m^k} 
		\hspace{-0.06in} \left( \hspace{-0.06in}
		\begin{array}{c} 
			|\mathcal{U}_m^k| \\
			i
		\end{array}
		\hspace{-0.06in} \right)
		(1-\delta_m)^{|\mathcal{U}_m^k|-i}(\delta_m)^i, \hspace{-0.05in}
\end{eqnarray}
where set $\psi_m^k$ is defined as:
\begin{equation}
  \psi_m^k = \left\{ i \left| \left[ 1 + \varphi_m^i\phi_m^{ | \mathcal{U}_m^k|-i} 
	     \frac{\Pr(H_1|\pi_m^k)}{\Pr(H_0|\pi_m^k)} \right]^{-1}  \ge \kappa_m^k \right. \right\}.
\end{equation}
For non-intrusive spectrum access, the collision probability should be bounded with a prescribed threshold $\gamma_m^k$. 
A higher spectrum access threshold $\kappa_m^k$ will reduce the potential interference with primary users, but increase the chance of wasting transmission opportunities.  For a given collision tolerance $\gamma_m^k$, we can solve $\Pr(A_m^k=0|H_1) = \gamma_m^k$ for $\kappa_m^k$.  The objective is to maximize CR users' spectrum access without exceeding the maximum collision probability with primary users. 

Let $\Omega_{i,j}$ be the set of available channels at link $\{i,j\}$. Assuming $i \in \mathcal{U}^k$ and $j \in \mathcal{U}^{k'}$, we have
\begin{eqnarray}\label{eq:DefChi}
  \Omega_{i,j} = \left\{ m \left| A_m^k=0 \mbox{ and } A_m^{k'}=0 \right. \right\}. 
\end{eqnarray}

%-------------------------------------------------------------------
\subsubsection{Link and Path Statistics}
%-------------------------------------------------------------------

Due to the amplify-and-forward approach for video data transmission, there is no queueing delay at intermediate nodes.  Assume each link has a fixed delay $\omega_{i,j}$ (i.e., processing and propagation delays).  Let $\mathcal{P}_l^A$ be the set of all possible paths from $z_l$ to $d_l$.  For a given delay requirement $T_{th}$, the set of feasible paths $\mathcal{P}_l$ for video session $l$ can be determined as:
\begin{eqnarray}\label{eq:PathDelay}
  \mathcal{P}_l = \left\{ \mathcal{P} \left| \sum_{\{i,j\} \in \mathcal{P}}
      \; \omega_{i,j} \le T_{th}, \; \mathcal{P} \in \mathcal{P}_l^A \right. \right\}.
\end{eqnarray}
 
Let $p_{i,j}^m$ be the packet loss rate on channel $m$ at link $\{i,j\}$.  A packet is successfully delivered over link $\{i,j\}$ if there is no loss on all the channels that were used for transmitting the packet.  The link loss probability $p_{i,j}$ can be derived as:
\begin{eqnarray}\label{eq:LinkLoss}
  p_{i,j} = 1 - \prod_{m \in \mathcal{M}} 
            \; (1-p_{i,j}^m)^{I_m},
\end{eqnarray}
where $\mathcal{M}$ is set of licensed channels and $I_m$ is an indicator: $I_m = 1$ if channel $m$ is used for the transmission, and $I_m = 0$ otherwise.  Assuming independent link losses, the end-to-end loss probability for path $\mathcal{P}_l^h \in \mathcal{P}_l$ can be estimated as:
\begin{eqnarray}\label{eq:PathLoss}
  p_l^h = 1 - \prod_{\{i,j\}\in \mathcal{P}_l^h} \; (1-p_{i,j}).
\end{eqnarray}

%-------------------------------------------------------------------
\subsection{Problem Statement}
%-------------------------------------------------------------------

We also aim to achieve fairness among the concurrent video sessions. It has been shown that {\em proportional fairness} can be achieved by maximizing the sum of logarithms of video PSNRs (i.e., utilities). Therefore, our objective is to maximize the overall system utility, i.e.,
\begin{equation}\label{eq:MaxQ}
  \mbox{maximize:} \;\; \sum_l U_l(R_l) = \sum_l \log (Q_l(R_l)). 
\end{equation}

%-------------------------------------------------------------------
\subsubsection{Multi-hop CR Network Video Streaming Problem}

The problem of video over multi-hop CR networks consists of path selection for each video session and channel scheduling for each CR node along the chosen paths.  We define two sets of index variables.  For channel scheduling, we have
\begin{eqnarray}\label{eq:DefX1}
  x_{i,j,m}^{l,h,r}=\left\{
    \begin{array}{ll}
      1, & \mbox{at link $\{i,j\}$, if channel $m$ is} \\
         & \mbox{assigned to tunnel $r$ in path $\mathcal{P}_l^h$} \\
      0, & \mbox{otherwise.}
    \end{array}\right.
\end{eqnarray}
For path selection, we have
\begin{eqnarray}\label{eq:DefY}
	y_l^h=\left\{
	\begin{array}{ll}
		1, & \mbox{if video session $l$ selects path $\mathcal{P}_l^h \in \mathcal{P}_l$} \\
		0, & \mbox{otherwise},
	\end{array}\right.
\end{eqnarray}

Note that the indicators, $x_{i,j,m}^{l,h,r}$ and $y_l^h$, are not independent.  If $y_l^h=0$ for path $\mathcal{P}_l^h$, all the $x_{i,j,m}^{l,h,r}$'s on that path are $0$.  If link $\{i,j\}$ is not on path $\mathcal{P}_l^h$, all its $x_{i,j,m}^{l,h,r}$'s are also $0$.   For link $\{i,j\}$ on path $\mathcal{P}_l^h$, we can only choose those available channels in set $\Omega_{i,j}$ to schedule video transmission. That is, we have $x_{i,j,m}^{l,h,r} \in \{0, 1\}$ if $m \in \Omega_{i,j}$, and $x_{i,j,m}^{l,h,r} = 0$ otherwise.  In the rest of the paper, we use $\mathbf{x}$ and $\mathbf{y}$ to represent the vector forms of $x_{i,j,m}^{l,h,r}$ and $y_l^h$, respectively.

As discussed, the objective is to maximize the expected utility sum at the end of $N_G$ time slots, as given in (\ref{eq:MaxQ}).  Since $\log(Q_l(\mbox{E}[R_l(0)]))$ is a constant, (\ref{eq:MaxQ}) is equivalent to the sum of utility increments of all the time slots, as 
\begin{eqnarray}\label{eq:ObjFunPrf12}
  && \hspace{-0.3in} \sum_l \log(Q_l(\mbox{E}[R_l(N_G)])) - \log(Q_l(\mbox{E}[R_l(0)])) \nonumber \\
  && \hspace{-0.45in} = \sum_t \sum_l \left\{ \log(Q_l(\mbox{E}[R_l(t)])) - \log(Q_l(\mbox{E}[R_l(t-1)])) \right\}. 
\end{eqnarray}
Therefore, (\ref{eq:MaxQ}) will be maximized if we maximize the expected utility increment during each time slot, which can be written as:
\begin{eqnarray} \label{eq:objrewritten}
	 & & \hspace{-0.1in} \sum_l \log(Q_l(\mbox{E}[R_l(t)])) -
	    \log(Q_l(\mbox{E}[R_l(t-1)])) 
	 \nonumber \\
	 &=& 
	 \sum_l \log \left( 1 +  \beta_l \frac{ \mbox{E}[R_l(t)] - \mbox{E}[R_l(t-1)] }{ Q_l(\mbox{E}[R_l(t-1)])} \right) \nonumber \\
	 &=&  
	 \sum_l\sum_{h \in \mathcal{P}_l} y_l^h \log \hspace{-0.025in} \left( \hspace{-0.025in} 1 \hspace{-0.025in} + \hspace{-0.025in} \sum_r \sum_m \frac{\beta_l L_p x_{z_l, z'_l, m}^{l,h,r}}{N_G T_s Q_l^{t-1}} (1-p_{l,h}^r) \hspace{-0.025in} \right) 
	 \nonumber \\
	 &=& 
	 \sum_l\sum_{h \in \mathcal{P}_l}y_l^h \log \hspace{-0.025in} \left( \hspace{-0.025in} 1 \hspace{-0.025in} + \hspace{-0.025in} \rho_l^t \sum_r \sum_m x_{z_l,z'_l,m}^{l,h,r} (1-p_{l,h}^r) \hspace{-0.025in} \right), \nonumber 
\end{eqnarray}
where $z'_l$ is the next hop from $z_l$ on path $\mathcal{P}_l^h$, $p_{l,h}^r$ is the packet loss rate on tunnel $r$ of path $\mathcal{P}_l^h$, $Q_l^{t-1} = Q_l(\mbox{E}[R_l(t-1)])$, and $\rho_l^t = \beta_l L_p / (N_G T_s Q_l^{t-1})$. 

From (\ref{eq:LinkLoss}) and (\ref{eq:PathLoss}), the end-to-end packet loss rate for tunnel $r$ on path $\mathcal{P}_l^h$ is:
\begin{equation}\label{eq:StrmLoss}
  p_{l,h}^r = 1 - \prod_{\{i,j\} \in \mathcal{P}_l^h} 
          \prod_{m \in \mathcal{M}}
          (1-p_{i,j}^m)^{x_{i,j,m}^{l,h,r}} .
\end{equation}
We assume that each tunnel can only include one channel on each link. When there are multiple channels available at each link along the path, a CR source node can set up multiple tunnels to exploit the additional bandwidth. We then have the following constraint:
\begin{equation}\label{eq:OneChan1}
  \sum_m x_{i,j,m}^{l,h,r} \le 1, \;\; \forall \; \{i,j\} \in \mathcal{P}_l^h.
\end{equation}
Considering availability of the channels, we further have,
\begin{equation}\label{eq:AvailChan}
  \sum_r\sum_m x_{i,j,m}^{l,h,r} \le |\Omega_{i,j}|, \;\; \forall \; \{i,j\} \in \mathcal{P}_l^h,
\end{equation}
where $|\Omega_{i,j}|$ is the number of available channels on link $\{i,j\}$ defined in (\ref{eq:DefChi}).

As discussed, each node is equipped with two transceivers: one for receiving and the other for transmitting video data during the transmission phase.  Hence a channel cannot be used to receive and transmit data simultaneously at a relay node.  We have for each channel $m$:
\begin{eqnarray}\label{eq:HalfDuplex}
  %& & \hspace{-0.4in} 
  \sum_r x_{i,j,m}^{l,h,r} + \sum_r x_{j,k,m}^{l,h,r} \le 1, \;\; 
  %\nonumber \\
  %& & \hspace{0.5in} 
  \forall \; m, l, \forall \; h \in \mathcal{P}_l, \forall \; \{i,j\} ,\{j,k\} \in \mathcal{P}_l^h. 
\end{eqnarray}

Let $n_l^h$ be the number of tunnels on path $\mathcal{P}_l^h$.  For each source $z_l$ and each destination $d_l$, the number of scheduled channels is equal to $n_l^h$.  We have for each source node
\begin{equation}\label{eq:SrcDstFlow}
  \sum_r\sum_m x_{z_l,z'_l,m}^{l,h,r}=n_l^h y_l^h, \;\; \forall \; h \in \mathcal{P}_l, \forall \; l.
\end{equation}
Let $d'_l$ be the last hop to destination $d_l$ on path $\mathcal{P}_l^h$, we have for each destination node
\begin{equation}\label{eq:DstFlow}
  \sum_r\sum_m x_{d'_l,d_l,m}^{l,h,r}=n_l^h y_l^h, \;\; \forall \; h \in \mathcal{P}_l, \forall \; l.
\end{equation}

At a relay node, the number of channels used to receive data is equal to that of channels used to transmit data, due to flow conservation and amplify-and-forward.  At relay node $j$ for session $l$, assume $\{i,j\}\in \mathcal{P}_l^h$ and $\{j,k\}\in \mathcal{P}_l^h$.  We have,
\begin{eqnarray}\label{eq:IntmFlow}
  %& & \hspace{-0.4in}  
  \sum_r\sum_m x_{i,j,m}^{l,h,r} = \sum_r\sum_m x_{j,k,m}^{l,h,r}, \;\; 
  %\nonumber \\
  %& & \hspace{0.35in} 
  \forall \; h \in \mathcal{P}_l, \forall \; l, \forall \; \{i,j\} ,\{j,k\} \in \mathcal{P}_l^h.
\end{eqnarray}

We also consider hardware-related constraints on path selection. We summarize such constraints in the following general form for ease of presentation:
\begin{equation} \label{eq:PathCon}
  \sum_l\sum_{h\in \mathcal{P}_l} w_{l,h}^g y_l^h \le 1, \forall \; g.
\end{equation} 
To simplify exposition, we choose at most one path in $\mathcal{P}_l$ for video session $l$.  Such a single path routing constraint can be expressed as $\sum_h y_l^h \leq 1$, which is a special case of (\ref{eq:PathCon}) where $w_{l,h}^1=1$ for all $h$, and $w_{l',h}^g=0$ for all $g \neq 1$, $l' \neq l$, and $h$.  We can also have $\sum_h y_l^h \leq \xi$ to allow up to $\xi$ paths for each video session.  In order to achieve optimality in the general case of multi-path routing, an optimal scheduling algorithm should be designed to dispatch packets to paths with different conditions (e.g., different number of tunnels and delays). 

There are also disjointedness constraints for the chosen paths.  This is because each CR node is equipped with two transceivers and both will be used for a video session if it is included in a chosen path.  Such disjointedness constraint is also a special case of (\ref{eq:PathCon}) with the following definition for $w_{l,h}^g$ for each CR node $g$: 
\begin{equation} \label{eq:d}
  w_{l,h}^g = \left\{ \begin{array}{ll}
          1, & \mbox{if node $g \in $ path $\mathcal{P}_l^h$} \\
          0, & \mbox{otherwise},
                      \end{array} \right. 
\end{equation}

Finally we formulate the problem of multi-hop CR network video streaming (OPT-CRV) as:
\begin{eqnarray}
	&& \hspace{-0.4in} \mbox{max:} \;\; \sum_l \hspace{-0.03in} \sum_{h \in \mathcal{P}_l}y_l^h \log \hspace{-0.03in} \left( \hspace{-0.03in} 1 \hspace{-0.01in} + \hspace{-0.01in} \rho_l^t \sum_r \sum_m x_{z_l,z'_l,m}^{l,h,r} (1 \hspace{-0.01in} - \hspace{-0.01in} p_{l,h}^r) \hspace{-0.03in} \right) \label{eq:PriProb} \\
	&& \hspace{-0.4in} \mbox{subject to:} \;\;\; (\ref{eq:DefX1}) \sim (\ref{eq:PathCon}).  \nonumber
\end{eqnarray}

%%%%%%%%%%%%%%%%%%%%%%%%%%%%%%%%%%%%%%%%%%%%%%%%%%%%%%%%%%%%%%%
\subsubsection{Centralized Algorithm and Upper/Lower Bounds}

Problem OPT-CRV is in the form of MINLP (without continuous variables), which is NP-hard in general.  We first describe a centralized algorithm to derive performance bounds in this section, and then present a distributed algorithm based on dual decomposition in the next section.  

We first obtain a relaxed {\em non-linear programming} (NLP) version of OPT-CRV.  The binary variables $x_{i,j,m}^{l,h,r}$ and $y_l^h$ are relaxed to take values in [0,1].  The integer variables $n_l^h$ are treated as nonnegative real numbers.  It can be shown that the relaxed problem has a concave object function and the constraints are convex.  This relaxed problem can be solved using a constrained nonlinear optimization problem solver.  If all the variables are integer in the solution, then we have the exact optimal solution.  Otherwise, we obtain an infeasible solution, which produces an upper bound for the problem. This is given in Lines 1$\sim$2 in Table~\ref{tab:CentrialAlgo}. 

We also develop a {\em sequential fixing algorithm} (SF) for solving OPT-CRV.  The pseudo-code is given in Table~\ref{tab:CentrialAlgo}.  SF iteratively solves the relaxed problem, fixing one or more integer variables after each iteration~\cite{Hou08}.  In Table~\ref{tab:CentrialAlgo}, Lines 3$\sim$7 fix the path selection variables $y_l^h$, and Lines 8$\sim$16 fix the channel scheduling variables $x_{i,j,m}^{l,h,r}$ and tunnel variables $n_l^h$.  The tunnel variables $n_l^h$ can be computed using (\ref{eq:SrcDstFlow}) after $x_{i,j,m}^{l,h,r}$ and $y_l^h$ are solved.  When the algorithm terminates, it produces a feasible solution that yields a lower bound for the objective value.

%------------------------------------------------------------
\begin{table}[!t]
\begin{center}
\caption{The Sequential Fixing Algorithm (SF) for Problem OPT-CRV}
\begin{tabular}{ll}
\hline
1 : & Relax integer variables $x_{i,j,m}^{l,h,r}$, $y_l^h$, and $n_l^h$; \\ 
2 : & Solve the relaxed problem using a constrained NLP solver; \\ 
3 : & {\bf if} (there is $y_l^h$ not fixed) \\
4 : & $\;\;$ Find the largest $y_{l'}^{h'}$, where $[l',h']=\arg \max\{y_l^h\}$, \\
    & $\;\;$ and fix it to $1$; \\
5 : & $\;\;$ Fix other $y_l^h$'s according to constraint (\ref{eq:PathCon}); \\
6 : & $\;\;$ Go to Step 2; \\ 
7 : & {\bf end if} \\
8 : & {\bf if} (there is $x_{i,j,m}^{l,h,r}$ not fixed) \\
9 : & $\;\;$ Find the largest $x_{i',j',m'}^{l',h',r'}$, where $[i',j',m',l',h',r']=$ \\
    & $\;\;$ $\arg \max\{x_{i,j,m}^{l,h,r}\}$, and set it to 1; \\
10: & $\;\;$ Fix other $x_{i,j,m}^{l,h,r}$'s according to the constraints; \\
11: & $\;\;$ {\bf if} (there is other variable that is not fixed) \\
12: & $\;\;\;\;$ Go to Step 2; \\
13: & $\;\;$ {\bf else} \\
14: & $\;\;\;\;$ Fix $n_l^h$'s based on $\mathbf{x}$ and $\mathbf{y}$; \\
15: & $\;\;\;\;$ Exit with feasible solution $\{\mathbf{x}, \mathbf{y},\mathbf{n}\}$; \\
16: & $\;\;$ {\bf end if} \\
17: & {\bf end if} \\
\hline
\end{tabular}
\label{tab:CentrialAlgo}
\end{center}
\end{table}
%--------------------------------------------------------------

%%%%%%%%%%%%%%%%%%%%%%%%%%%%%%%%%%%%%%%%%%%%%%%%%%%%%%%%%%%%%%%
\subsection{Dual Decomposition \label{sec:dual}}
%%%%%%%%%%%%%%%%%%%%%%%%%%%%%%%%%%%%%%%%%%%%%%%%%%%%%%%%%%%%%%%

SF is a centralized algorithm requiring global information.  It may not be suitable for multi-hop wireless networks, although the upper and lower bounds provide useful insights on the performance limits. In this section, we develop a distributed algorithm for Problem OPT-CRV and analyze its optimality and convergence performance.

%%%%%%%%%%%%%%%%%%%%%%%%%%%%%%%%%%%%%%%%%%%%%%%%%%%%%%%%%%%%%%%
\subsubsection{Decompose Problem OPT-CRV}

Since the domains of $x_{i,j,m}^{l,h,r}$ defined in (\ref{eq:OneChan1})$\sim$(\ref{eq:IntmFlow}) for different paths do not intersect with each other, we can decompose Problem OPT-CRV into two subproblems.  The first subproblem deals with channel scheduling for maximizing the expected utility on a chosen path $\mathcal{P}_l^h$.  We have the {\em channel scheduling} problem (OPT-CS) as: 
\begin{eqnarray}
	&& \hspace{-0.2in} H_l^h=\max_{\mathbf{x}} \sum_r \sum_m
      x_{z_l,z'_l,m}^{l,h,r}(1-p_{l,h}^r)  \label{eq:optcs} \\
&& \hspace{-0.2in} \mbox{subject to:} \; (\ref{eq:OneChan1})\sim(\ref{eq:IntmFlow}), \; x_{z_l,z'_l,m}^{l,h,r} \in \{0,1\}, \; \mbox{for all } l, h, r, m. \nonumber
\end{eqnarray}
In the second part, optimal paths are selected to maximize the overall objective function.  Letting $F_l^h = \log \left( 1 + \rho_l^T H_l^h \right)$, we have the following {\em path selection} problem (OPT-PS):
\begin{eqnarray}
	\mbox{maximize:} && f(\mathbf{y}) = 
    \sum_l \sum_h F_l^h y_l^h  \label{eq:PathSelec} \\
	  \mbox{subject to:} && \sum_l\sum_{h 
	           \in \mathcal{P}_l} 
	  w_{l,h}^g y_l^h \le 1, \; \mbox{for all } \; g \nonumber \\
	  && y_l^h \in \{0,1\}, \; \mbox{for all } \; l, h. \nonumber
\end{eqnarray}

%%%%%%%%%%%%%%%%%%%%%%%%%%%%%%%%%%%%%%%%%%%%%%%%%%%%%%%%%%%%%%
\subsubsection{Solve the Channel Scheduling Subproblem}

We have the following result for assigning available channels at a relay node. 

\smallskip
\begin{theorem} \label{th:th1}
Consider three consecutive nodes along a path, denoted as nodes $i$, $j$, and $k$.  Idle channels 1 and 2 are available at link $\{i, j\}$ and idle channels 3 and 4 are available at link $\{j, k\}$.  Assume the packet loss rates of the four channels satisfy $p_{i,j}^1 > p_{i,j}^2$ and $p_{j,k}^3 > p_{j,k}^4$.  To set up two tunnels, assigning channels \{1, 3\} to one tunnel and channels \{2, 4\} to the other tunnel achieves the maximum expectation of successful transmission on path section $\{i, j, k\}$.
\end{theorem}
\begin{proof}
Let the success probabilities on the channels be $\tilde{p}_{i,j}^1=1-p_{i,j}^1$, $\tilde{p}_{i,j}^2=1-p_{i,j}^2$, $\tilde{p}_{j,k}^3=1-p_{j,k}^3$, and $\tilde{p}_{j,k}^4=1-p_{j,k}^4$.  We have $\tilde{p}_{i,j}^1 < \tilde{p}_{i,j}^2$ and $\tilde{p}_{j,k}^3<\tilde{p}_{j,k}^4$. 
Comparing the success probabilities of the channel assignment given in Theorem~\ref{th:th1} and that of the alternative assignment, we have
$\tilde{p}_{i,j}^1 \tilde{p}_{j,k}^3 + \tilde{p}_{i,j}^2 \tilde{p}_{j,k}^4 -
	 \tilde{p}_{i,j}^1 \tilde{p}_{j,k}^4 - \tilde{p}_{i,j}^2 \tilde{p}_{j,k}^3 
=(\tilde{p}_{i,j}^1-\tilde{p}_{i,j}^2)(\tilde{p}_{j,k}^3-\tilde{p}_{j,k}^4) > 0$.
The result follows.
\end{proof}
\smallskip

According to Theorem~\ref{th:th1}, a greedy approach, which always chooses the channel with the lowest loss rate at each link when setting up tunnels along a path, produces the optimal overall success probability.  More specifically, when there is only one tunnel to be set up along a path, the tunnel should consist of the most reliable channels available at each link along the path.  When there are multiple tunnels to set up along a path, tunnel 1 should consist of the most reliable channels that are available at each link; tunnel 2 should consist of the second most reliable links available at each link; and so forth. 
 
Define the set of loss rates of the available channels on link $\{i, j\}$ as $\Lambda_{i,j}=\{p_{i,j}^m|m \in \Omega_{i,j} \}$.  The greedy algorithm is given in Table~\ref{tab:ChanSchedAlgo}, with which each video source node solves Problem OPT-CS for each feasible path.  Lines 2$\sim$3 in Table~\ref{tab:ChanSchedAlgo} checks if there is more channels to assign and the algorithm terminates if no channel is left.  In Lines 4$\sim$10, links with only one available channel are assigned to tunnel $r$ and the neighboring links with the same available channels are removed due to constraint~(\ref{eq:HalfDuplex}).  In Lines 11$\sim$17, links with more than two channels are grouped to be assigned later.  In Lines 18$\sim$20, the available channel with the lowest packet loss rate is assigned to tunnel $r$ at each unallocated link, according to Theorem~\ref{th:th1}.  To avoid co-channel interference, the same channel on neighboring links is removed as in Lines 21$\sim$33.

%-------------------------------------------------------------
\begin{table}[!t]
\begin{center}
\caption{The Greedy Algorithm for Channel Scheduling}
\begin{tabular}{ll}
\hline
1 : & Initialization: tunnel $r=1$, link $\{i,j\}$'s from $z_l$ to $d_l$; \\ 
2 : & {\bf if} ($|\Lambda_{i,j}|$ == 0)  \\
3 : & $\;\;$ Exit; \\
4 : & {\bf else if} ($|\Lambda_{i,j}|$ == 1) \\
5 : & $\;\;$ Assign the single channel in $\Lambda_{i,j}$, $m'$, to tunnel $r$;  \\
6    & $\;\;$ Check neighboring link $\{k,i\}$; \\ 
7 : & $\;\;$ {\bf if} ($p_{k,i}^{m'} \in \Lambda_{k,i}$) \\
8 : & $\;\;\;\;$ Remove $p_{k,i}^{m'}$ from $\Lambda_{k,i}$,  \\
    & $\;\;\;\;$ $i \leftarrow k$, $j\leftarrow i$ and go to Step 2; \\
9 : & $\;\;$ {\bf else}  \\
10: & $\;\;\;\;$  Go to Step 13; \\
11: & $\;\;$ {\bf end if} \\
12: & {\bf else}  \\
13: & $\;\;$ Put $\Lambda_{i,j}$ in set $\Lambda_l^h$; \\
14: & $\;\;$ {\bf if} (node $j$ is not destination $d_l$) \\
15: & $\;\;\;\;$ $i \leftarrow j$, $j \leftarrow v$; \\
16: & $\;\;\;\;$ Go to Step 2; \\
17: & $\;\;$ {\bf end if} \\
18: & {\bf end if} \\
19: & {\bf while} ($\Lambda_l^h$ is not empty) \\
20: & $\;\;$ Find the maximum value $p_{i',j'}^{m'}$ in set $\Lambda_l^h$ \\
    & $\;\;$ $\{i',j',m'\}=\arg\min \{p_{i,j}^m\}$; \\
21: & $\;\;$ Assign channel $m'$ to tunnel $r$; \\
22: & $\;\;$ Remove set $\Lambda_{i',j'}$ from set $\Lambda_l^h$; \\
23: & $\;\;$ Check neighboring link $\{k,i\}$ and $\{j,v\}$; \\
24: & $\;\;$ {\bf if} ($p_{k,i}^{m'} \in \Lambda_{k,i}$ and $\Lambda_{k,i} \in \Lambda_l^h$) \\
25: & $\;\;\;\;$ Remove $p_{k,i}^{m'}$ from $\Lambda_{k,i}$; \\
26: & $\;\;\;\;$ {\bf if} ($\Lambda_{k,i}$ is empty) \\
27: & $\;\;\;\;\;\;$ Exit; \\
28: & $\;\;\;\;$ {\bf end if} \\
29: & $\;\;$ {\bf end if} \\
30: & $\;\;$ {\bf if} ($p_{j,v}^{m'} \in \Lambda_{j,v}$ and $\Lambda_{j,v} \in \Lambda_l^h$) \\
31: & $\;\;\;\;$ Remove $p_{j,v}^{m'}$ from $\Lambda_{j,v}$; \\
32: & $\;\;\;\;$ {\bf if} ($\Lambda_{j,v}$ is empty) \\
33: & $\;\;\;\;\;\;$ Exit; \\
34: & $\;\;\;\;$ {\bf end if} \\
35: & $\;\;$ {\bf end if} \\
36: & {\bf end while} \\
37: & Compute the next tunnel: $r \leftarrow r+1$ and go to Step 2; \\
\hline
\end{tabular}
\label{tab:ChanSchedAlgo}
\end{center}
\end{table}
%------------------------------------------------------------

%%%%%%%%%%%%%%%%%%%%%%%%%%%%%%%%%%%%%%%%%%%%%%%%%%%%%%%%%%%%%%
\subsubsection{Solve the Path Selection Subproblem} 

To solve Problem OPT-PS, we first relax binary variables $y_l^h$ to allow them take real values in [0,1] and obtain the following {\em relaxed path selection} problem (OPT-rPS):
\begin{eqnarray}
	\mbox{maximize:} &&  f(\mathbf{y}) = \sum_l \sum_h F_l^h
	                      y_l^h \label{eq:rPathSelec} \\
	\mbox{subject to:} && \sum_l\sum_{h\in \mathcal{P}_l} 
	                      w_{l,h}^g y_l^h \le 1, \; \mbox{for all } \; g \nonumber \\
	                   && 0 \leq y_l^h \leq 1, \; \mbox{for all } h,l.   \nonumber
\end{eqnarray}
We then introduce positive Lagrange Multipliers $e_g$ for the path selection constraints in Problem OPT-rPS and obtain the corresponding {\em Lagrangian function}:
\begin{eqnarray} 
	&& \hspace{-0.45in} \mathcal{L}(\mathbf{y}, \mathbf{e}) = \sum_l\sum_h F_l^h y_l^h + \sum_g e_g (1 - \sum_l\sum_h w_{l,h}^g y_l^h) 
	\label{eq:LagFun} \\
	&& = \sum_l\sum_h ( F_l^hy_l^h - \sum_g w_{l,h}^g y_l^h e_g ) + \sum_g e_g \nonumber \\
	&& =	\sum_l\sum_h \mathcal{L}_l^h ( y_l^h, \mathbf{e} )+ \sum_g e_g.  \nonumber
\end{eqnarray}
Problem (\ref{eq:LagFun}) can be decoupled since the domains of $y_l^h$'s do not overlap.  Relaxing the coupling constraints, it can be decomposed into two levels.  At the lower level, we have the following subproblems, one for each path $\mathcal{P}_l^h$, 
\begin{eqnarray}\label{eq:FirstLevel}
   \max_{0 \leq y_l^h \leq 1} \mathcal{L}_l^h(y_l^h, \mathbf{e})=F_l^hy_l^h - \sum_g w_{l,h}^g y_l^h e_g.
\end{eqnarray}
At the higher level, by updating the dual variables $e_g$, we can solve the {\em relaxed dual problem}:
\begin{eqnarray}\label{eq:SecondLevel}
  \min_{\mathbf{e} \ge 0} \;\; q(\mathbf{e}) = \sum_l\sum_h \mathcal{L}_l^h \left( \left( y_l^h \right)^{\ast}, \mathbf{e} \right) + \sum_g e_g,
\end{eqnarray}
where $\left( y_l^h \right)^{\ast}$ is the optimal solution to (\ref{eq:FirstLevel}).  Since the solution to (\ref{eq:FirstLevel}) is unique, the relaxed dual problem (\ref{eq:SecondLevel}) can be solved using the following {\em subgradient method} that iteratively updates the Lagrange Multipliers~\cite{Bertsekas99}:
\begin{eqnarray}\label{eq:SubGrad}
  e_g(\tau+1) = \left[ e_g(\tau)-\alpha(\tau)(1 - \sum_l\sum_h w_{l,h}^g y_l^h) \right]^+,
\end{eqnarray}
where $\tau$ is the iteration index, $\alpha(\tau)$ is a sufficiently small positive step size and $[x]^+$ denotes $\max\{x, 0\}$.  The pseudo code for the distributed algorithm is given in Table~\ref{tab:PathSelecAlgo}.

%--------------------------------------------------------------
\begin{table}[!t]
\begin{center}
\caption{Distribution Algorithm for Path Selection}
\begin{tabular}{ll}
\hline
1: & Initialization: set $\tau=0$, $e_g(0) >0 $ and step size $s \in [0,1]$; \\ 
2: & Each source locally solves the lower level problem in (\ref{eq:FirstLevel}); \\
   & {\bf if} ($F_l^h-\sum_g d_{l,h}^g e_g(\tau)) > 0$) \ $y_l^h=y_l^h+s$, $y_l^h= \min \{y_l^h,1\}$;\\ 
   & {\bf else} \ \ $y_l^h=y_l^h-s$, $y_l^h=\max\{y_l^h,0\}$; \\
3: & Broadcast solution $y_l^h(\mathbf{e}(\tau))$; \\
4: & Each source updates $\mathbf{e}$ according to (\ref{eq:SubGrad}) and broadcasts 
        $\mathbf{e}(\tau+1)$ \\
   & through the common control channel; \\
5: & $\tau \leftarrow \tau$+1 and go to Step 2 until termination criterion is satisfied;
   \\ \hline
\end{tabular}
\label{tab:PathSelecAlgo}
\end{center}
\end{table}
%--------------------------------------------------------------

%%%%%%%%%%%%%%%%%%%%%%%%%%%%%%%%%%%%%%%%%%%%%%%%%%%%%%%%%%%%%%%
\subsubsection{Optimality and Convergence Analysis} 

The distributed algorithm in Table~\ref{tab:PathSelecAlgo} iteratively updates the dual variables until they converge to stable values.  In this section, we first prove that the solution obtained by the distributed algorithm is also optimal for the original path selection problem OPT-PS.  We then derive the convergence condition for the distributed algorithm. 

\begin{fact}[\cite{Bertsekas99}]  \label{thm:EqualCont}
Consider a linear problem involving both equality and inequality constraints
\begin{eqnarray}
	\mbox{maximize:} && \mathbf{a}' \mathbf{x} \label{eq:InequalLP} \\
	\mbox{subject to:} && \mathbf{h}'_1\mathbf{x}=b_1, \; \cdots, \; 
        \mathbf{h}'_m\mathbf{x}=b_m   \nonumber \\
     && \mathbf{g}'_1\mathbf{x}\le c_1, \; \cdots, \; 
        \mathbf{g}'_r\mathbf{x} \le c_r, \nonumber
\end{eqnarray}
where $\mathbf{a}$, $\mathbf{h}_i$, and $\mathbf{g}_j$ are column vectors in $\mathcal{R}_n$, $b_i$'s and $c_j$'s are scalars, and $\mathbf{a}'$ is the {\em transpose} of $\mathbf{a}$. For any feasible point $\mathbf{x}$, the set of {\em active} inequality constraints is denoted by
  $\mathcal{A}(\mathbf{x}) = \left\{ j | \mathbf{g}'_j \mathbf{x}=c_j \right\}$.
If $\mathbf{x}^{\ast}$ is a maximizer of inequality constrained problem (\ref{eq:InequalLP}), $\mathbf{x}^{\ast}$ is also a maximizer of the following equality constrained problem:
\begin{eqnarray}
	\mbox{maximize:} && \mathbf{a}' \mathbf{x} \label{eq:eq:EqualLP} \\
	\mbox{subject to:} && \mathbf{h}'_1 \mathbf{x} = b_1, \; \cdots, \; 
  \mathbf{h}'_m \mathbf{x}=b_m \nonumber \\
  && \mathbf{g}'_j\mathbf{x} = c_j, \forall \; j \in \mathcal{A}(\mathbf{x}). \nonumber
\end{eqnarray} 
\end{fact}

\smallskip
\begin{lemma} \label{lemma:1}
The optimal solution for the relaxed primal problem OPT-rPS in (\ref{eq:rPathSelec}) is also feasible and optimal for the original Problem OPT-PS in (\ref{eq:PathSelec}).
\end{lemma}
\begin{proof}
According to Fact~\ref{thm:EqualCont}, the linearized problem of OPT-PS, i.e., OPT-rPS, can be rewritten as an equality constrained problem in the following form:
\begin{eqnarray}
	\mbox{maximize:} && \mathbf{F}'\mathbf{y} \label{eq:EqualCons} \\
	\mbox{subject to:} && \mathbf{w}_j' \mathbf{y}=1, \;\; j \in
     \mathcal{A}(\mathbf{y}^{\ast}) \label{eq:EqualCons1} \\
     && 0 \leq y_l^h \leq 1, \; \mbox{for all } h,l,  \nonumber
\end{eqnarray}
where $\mathbf{F}$, $\mathbf{w}_j$'s, and $\mathbf{y}$ are column vectors with elements $F_l^h$, $w_{l,h}^g$, and $y_l^h$, respectively.  We apply {\em Gauss-Jordan elimination} to the constraints in (\ref{eq:EqualCons1}) to solve for $\mathbf{y}$. Since there is not sufficient number of equations, some $y_l^h$'s are free variables (denoted as $y_i^f$) and the rest are dependent variables (denoted as $y_j^d$). Assuming there are $r$ free variables, the dependent variables can be written as linear combinations of the free variables after Gauss-Jordan elimination, as 
\begin{equation}\label{eq:depdent}
  y_j^d = \sum_{i=1}^r \bar{w}_j^i y_i^f + \bar{b}_j, 
          \; j \in \mathcal{A}(y_i^{\ast}).
\end{equation}
Due to Gauss-Jordan elimination and binary vectors $\mathbf{w}_j$'s, $\bar{w}_j^i$ and $\bar{b}_j$ in (\ref{eq:depdent}) are all integers.  Therefore, if all the free variables $y_i^f$ attain binary values, then all the dependent variables $y_j^d$ computed using (\ref{eq:depdent}) will also be integers.  Since $0 \leq y_j^d \leq 1$, being integers means that they are either 0 or 1, i.e., binaries. That is, such a solution will be feasible. 

Next we substitute (\ref{eq:depdent}) into problem (\ref{eq:EqualCons}) to eliminate all the dependent variables.   Then we obtain a unconstrained problem with only $r$ free variables, as
\begin{equation}
  \mbox{maximize:} \;\; \sum_{i=1}^r \bar{F}_i y_i^f + \bar{b}_0 
	\label{eq:GaussElimi}
\end{equation}
Since the free variables $y_i^f$'s take value in \{0, 1\}, this problem can be easily solved as follows. If the coefficient $\bar{F}_i > 0$, we set $y_i^f=1$; otherwise, if $\bar{F}_i < 0$, we set $y_i^f=0$.  Thus (\ref{eq:GaussElimi}) achieves its maximum objective value.  Once all the free variables are determined with their optimal binary values, we computes the dependent variables using (\ref{eq:depdent}), which are also binary as discussed above.  Thus we obtain a feasible solution, which is optimal. 
\end{proof}

\smallskip
\begin{lemma} \label{lemma:2}
If the relaxed primal Problem OPT-rPS in (\ref{eq:rPathSelec}) has an optimal solution, then the relaxed dual problem (\ref{eq:SecondLevel}) also has an optimal solution and the corresponding optimal values of the two problems are identical.
\end{lemma}
\begin{proof}
By definition, the problems in (\ref{eq:LagFun}) and (\ref{eq:SecondLevel}) are primal/dual problems.  The primal problem always has an optimal solution because it is bounded. Since Problem OPT-rPS is an LP problem, the relaxed dual problem is also bounded and feasible. Therefore the relaxed dual problem also has an optimal solution.  We have the {\em strong duality} if the primal problem is convex, which is the case here since Problem OPT-rPS is an LP problem. 
\end{proof}

We have Theorem~\ref{th:th2} on the optimality of the path selection solution, which follows naturally from Lemmas~\ref{lemma:1} and~\ref{lemma:2}.
\smallskip
\begin{theorem} \label{th:th2}
The optimal solution to the relaxed dual problem (\ref{eq:FirstLevel}) and (\ref{eq:SecondLevel}) is also feasible and optimal to the original path selection Problem OPT-PS given in (\ref{eq:PathSelec}). 
\end{theorem}

As discussed, the relaxed dual problem (\ref{eq:SecondLevel}) can be solved using the {\em subgradient method} that iteratively updates the Lagrange Multipliers.  We have the following theorem on the convergence of the distributed algorithm given in Table~\ref{tab:PathSelecAlgo}. 
\smallskip
\begin{theorem} 
Let $\mathbf{e}^{\ast}$ be the optimal solution.  The distributed algorithm in Table~\ref{tab:PathSelecAlgo} converges if the step sizes $\alpha(\tau)$ in~(\ref{eq:SubGrad}) satisfy the following condition:
\begin{eqnarray}\label{eq:StepSize}
  0 < \alpha(\tau) < \frac{2 \left[ q(\mathbf{e}(\tau))-q(\mathbf{e}^{\ast}) \;
       \right]}{||G(\tau)||^2}, \;\; \mbox{for all } \tau,
\end{eqnarray}
where $G(\tau)$ is the gradient of $q(\mathbf{e}(\tau))$. 
\end{theorem}

\begin{proof}
Since $q(\mathbf{e}(\tau))$ is a linear function, we have subgradient equality, as
\begin{eqnarray}\label{eq:SubGradIneq}
  q(\mathbf{e}(\tau))-q(\mathbf{e}^{\ast})= \left[ \mathbf{e}(\tau)-\mathbf{e}^{\ast} \right]'G(\tau). \nonumber
\end{eqnarray}
It then follows that
\begin{eqnarray}
	 && \hspace{-0.3in}	||\mathbf{e}(\tau)-\alpha(\tau)G(\tau)-\mathbf{e}^{\ast}||^2
\nonumber \\
 && \hspace{-0.3in} = ||\mathbf{e}(\tau) \hspace{-0.025in} - \hspace{-0.025in} \mathbf{e}^{\ast}||^2 - 2 \alpha(\tau) 
		     [ \mathbf{e}(\tau) \hspace{-0.025in} - \hspace{-0.025in} \mathbf{e}^{\ast} ]' G(\tau) 
		+ (\alpha(\tau))^2||G(\tau)||^2  \nonumber \\
		&& \hspace{-0.3in} = ||\mathbf{e}(\tau) \hspace{-0.025in} - \hspace{-0.025in} 
		    \mathbf{e}^{\ast}||^2-2\alpha(\tau) 
		    [ q(\mathbf{e}(\tau)) \hspace{-0.025in} - \hspace{-0.025in} q(\mathbf{e}^{\ast}) ] 
		+(\alpha(\tau))^2||G(\tau)||^2_{.} \nonumber \\ \label{eq:Inequal1}  
\end{eqnarray}
If $\alpha(\tau)$ satisfy (\ref{eq:StepSize}), the sum of the last two terms in (\ref{eq:Inequal1}) is negative. It follows that,
  $||\mathbf{e}(\tau)-\alpha(\tau)G(\tau)-\mathbf{e}^{\ast}|| 
  < ||\mathbf{e}(\tau)-\mathbf{e}^{\ast}||$. 
Since the projection operation is {\em nonexpansive}, we have,
\begin{eqnarray}\label{eq:Inequal3}
	&& \hspace{-0.3in} ||\mathbf{e}(\tau+1)-\mathbf{e}^{\ast}|| =  
||[\mathbf{e}(\tau)-\alpha(\tau)G(\tau)]^+-[\mathbf{e}^{\ast}]^+|| 
	  \nonumber \\
	  &\le&
	  ||\mathbf{e}(\tau)-\alpha(\tau)G(\tau)-\mathbf{e}^{\ast}||
	  < ||\mathbf{e}(\tau)-\mathbf{e}^{\ast}||, \nonumber
\end{eqnarray}
which states the conditional convergence of the algorithm. 
\end{proof}

Since the optimal solution $\mathbf{e}^{\ast}$ is not known a priori, we use the following approximation in the algorithm:
$\alpha(\tau)=\frac{q(\mathbf{e}(\tau))-\hat{q}(\tau)}{||G(\tau)||^2}$,
where $\hat{q}(\tau)$ is the current estimate for $q(\mathbf{e}^{\ast})$.  We choose the mean of the objective values of the relaxed primal and dual problems for $\hat{q}(\tau)$.

%%%%%%%%%%%%%%%%%%%%%%%%%%%%%%%%%%%%%%%%%%%%%%%%%%%%%%%%%%%%%%
\subsubsection{Practical Considerations} 

Our distributed algorithms are based on the fact that the computation is distributed on each feasible path.  The OPT-CS algorithm requires information on channel availability and packet loss rates at the links of feasible paths.  The OPT-PS algorithm computes the primal variable $y_l^h$ for each path and broadcasts Lagrangian multipliers over the control channel to all the source nodes.  We assume a perfect control channel such that channel information can be effectively distributed and shared, which is not confined by the time slot structure~\cite{Su08}.

We assume relatively large timescales for the primary network time slots, and small to medium diameter for the CR network, such that there is sufficient time for timely feedback of channel information to the video source nodes and for the distributed algorithms to converge.  Otherwise, channel information can be estimated using (\ref{eq:OneStep}) based on delayed feedback, leading to suboptimal solutions.  If the time slot is too short, the distributed algorithm may not converge to the optimal solution (see Fig.~\ref{fig:timeslot}).  We focus on developing the CR video framework in this paper, and will investigate these issues in our future work.

%--------------------------------------------------------------
\subsection{Simulation Results \label{sec:sim1}}
%--------------------------------------------------------------

%--------------------------------------------------------------
\subsubsection{Methodology and Simulation Settings} 

We implement the proposed algorithms with a combination of C and MATLAB (i.e., for solving the relaxed NLP problems), and evaluate their performance with simulations.  For the results reported in this section, we have $K=3$ primary networks and $M=10$ channels.  There are 56, 55, and 62 CR users in the coverage areas of primary networks 1, 2, and 3, respectively.  
The $\left| \mathcal{U}_m^1 \right|$'s are [5 4 6 4 8 7 5 6 7 4] (i.e., five users sense channel 1, four users sense channel 2, and so forth); the $\left| \mathcal{U}_m^2 \right|$'s are [4 6 5 7 6 5 3 8 5 6], and the $\left| \mathcal{U}_m^3 \right|$'s are [8 6 5 4 7 6 8 5 6 7]. The  topology is shown in Fig.~\ref{fig:topology}.  

%--------------------------------------------------------------
\begin{figure}[!t]
  \centering  
  \includegraphics [width=3.7in]{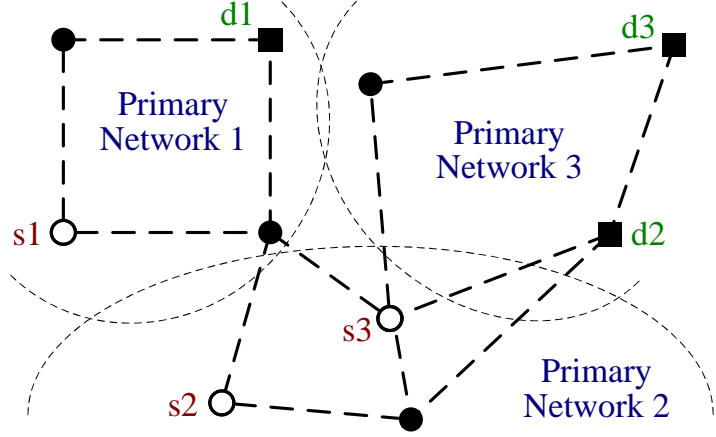}
  \caption{Topology of the multi-hop CR network. Note that only video source nodes, video destination nodes, and those nodes along the precomputed paths are shown in the topology. }
  \label{fig:topology}
\end{figure}
%--------------------------------------------------------------

We choose $L_p=100$, $T_s=0.02$ and $N_G=10$. The channel utilization is $\eta_m^k=0.6$ for all the channels. The probability of false alarm is $\epsilon_m^k=0.3$ and the probability of miss detection is $\delta_m^k=0.2$ for all $m$ and $k$, unless otherwise specified. Channel parameters $\lambda_m^k$ and $\mu_m^k$ are set between $(0,1)$. The maximum allowed collision probability $\gamma_m^k$ is set to $0.2$ for all the $M$ channels in the three primary networks.

We consider three video sessions, each streaming a video in the Common Intermediate Format (CIF, $352 \times 288$), i.e., {\em Bus} to destination 1, {\em Foreman} to destination 2, and {\em Mother \& Daughter} to destination 3.  The frame rate is 30 fps, and a GOP consists of 10 frames.  We assume that the duration of a time slot is 0.02 seconds and each GOP should be delivered in 0.2 seconds (i.e., 10 time slots).

We compare four schemes in the simulations: (i) the upper-bounding solution by solving the relaxed version of Problem OPT-CRV using an NLP solver, (ii) the proposed distributed algorithm in Tables~\ref{tab:ChanSchedAlgo} and~\ref{tab:PathSelecAlgo}, (iii) the sequential fixing algorithm given in Table~\ref{tab:CentrialAlgo}, which computes a lower-bounding solution, and (iv) a greedy heuristic where at each hop, the link with the most available channels is used.  Each point in the figures is the average of 10 simulation runs, with $95\%$ confidence intervals plotted as error bars in the figures.  The 95\% confidence intervals are negligible in all the figures.

%--------------------------------------------------------------
\subsubsection{Simulation Results} 

%--------------------------------------------------------------
\paragraph{Algorithm Performance} 

To demonstrate the convergence of the distributed algorithm, we plot the traces of the four Lagrangian multipliers in Fig.~\ref{fig:convergence}.  We observe that all the Lagrangian multipliers converge to their optimal values after 76 iterations.  We also plot the control overhead as measured by the number of distinct broadcast messages for $e_i(\tau)$ using the y-axis on the right-hand side.  The overhead curve increases linearly with the number of iterations and gets flat (i.e., no more broadcast message) when all the Lagrangian multipliers converge to their optimal values. 

We examine the impact of spectrum sensing errors in Fig.~\ref{fig:psnr-error}.  We test six sensing error combinations $\{\epsilon_m ,\delta_m\}$ as follows: \{0.1, 0.5\}, \{0.2, 0.3\}, \{0.3, 0.2\}, \{0.5, 0.11\}, \{0.7, 0.06\}, and \{0.9, 0.02\}, and plot the average PSNR values of the Foreman session. It is interesting to see that the best video quality is achieved when the false alarm probability $\epsilon_m$ is between $0.2$ and $0.3$. Since the two error probabilities are correlated, increasing one will generally decrease the other.  With a larger $\epsilon_m$, CR users are more likely to waste spectrum opportunities that are actually available, leading to lower bandwidth for videos and poorer video quality, as shown in Fig.~\ref{fig:psnr-error}.  On the other hand, a larger $\delta_m$ implies more aggressive spectrum access and more severe interference to primary users.  Therefore when $\epsilon_m$ is lower than 0.2 (and $\delta_m$ is higher than 0.3), the CR nodes themselves also suffer from the collisions and the video quality degrades. 

%--------------------------------------------------------------
\begin{figure}[!t]
  \centering  
  \includegraphics [width=5.0in, height=3.3in]{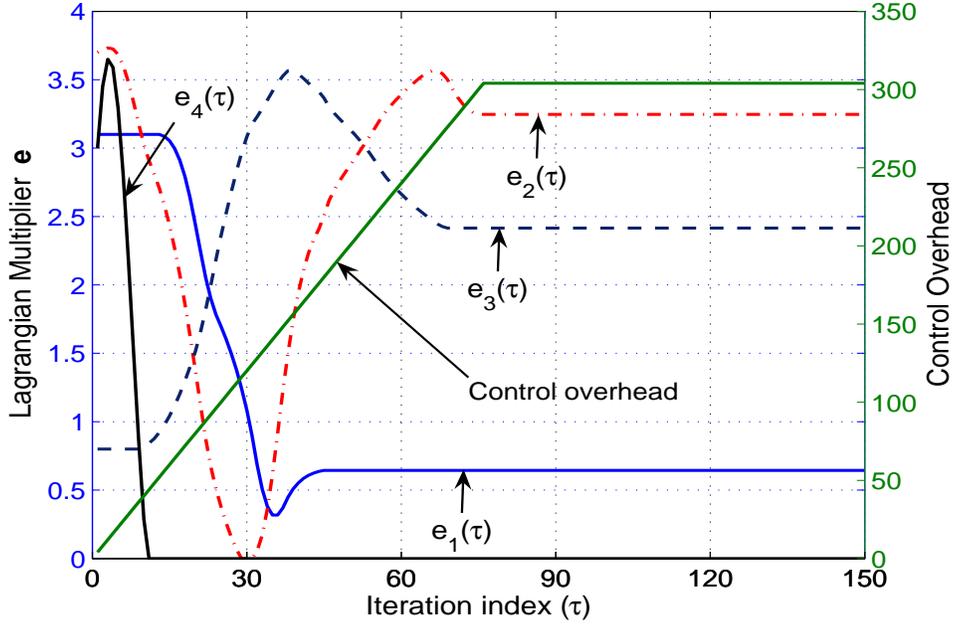}
  \caption{Illustrate the convergence of the distributed algorithm.}
  \label{fig:convergence}
\end{figure}
%-------------------------------------------------------------

%-------------------------------------------------------------
\begin{figure}[!t]
  \centering  
  \includegraphics [width=4.5in, height=3.0in]{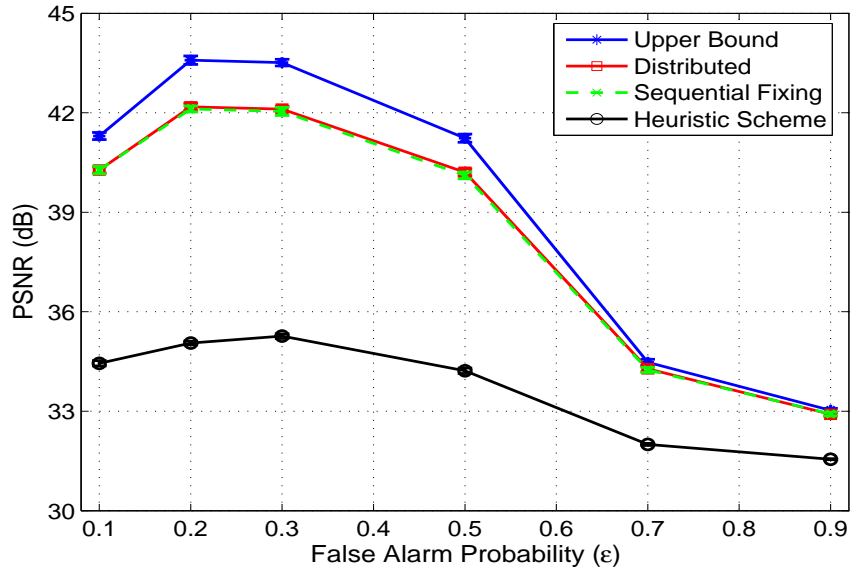}
  \caption{Video PSNRs versus spectrum sensing error.}
  \label{fig:psnr-error}
\end{figure}
%-------------------------------------------------------------

%--------------------------------------------------------------
\paragraph{Impact of Primary Network Parameters} 

In Fig.~\ref{fig:psnr-eta}, we examine the impact of channel utilization $\eta$ on received video quality.  We focus on Session 2 with the Foreman sequence. The average PSNRs achieved by the four schemes are plotted when $\eta$ is increased from $0.6$ to $0.9$ for all licensed channels. Intuitively, a smaller $\eta$ allows more transmission opportunities for CR nodes, leading to improved video quality. This is illustrated in the figure where all the four curves decrease as $\eta$ gets larger.  The distributed scheme achieves PSNRs very close to that obtained by sequential fixing, and both of them are close to the upper bound.  The heuristic scheme is inefficient in exploiting the available spectrum even when the channel utilization is low.  
As discussed, the time slot duration is also an important parameter that may affect the convergence of the distributed algorithm.  In Fig.~\ref{fig:timeslot}, we keep the same network and video session settings, while increasing the time slot duration as 4 ms, 10 ms, 20 ms, 40ms and 100 ms.  For a given time slot duration, we let the distributed algorithm run for 5\% of the time slot duration, starting from the beginning of the time slot, and then stop.  The solution that the algorithm produces when it is stopped will be used for video transmission in the remainder of this time slot.  It can be seen that when the time slot is 4 ms, the algorithm does not converge after 5\%$\times$4=0.2 ms, and the PSNR produced by the distributed algorithm is low (but still higher than that of the heuristic algorithm).  When the time slot duration is sufficiently large (e.g., over 10 ms), the algorithm can converge and the proposed algorithm produces very good video quality as compared to the upper bound and the lower bound given by the sequential fixing algorithm.   

%--------------------------------------------------------------
\begin{figure}[!t]
  \centering  
  \includegraphics [width=4.5in, height=3.0in]{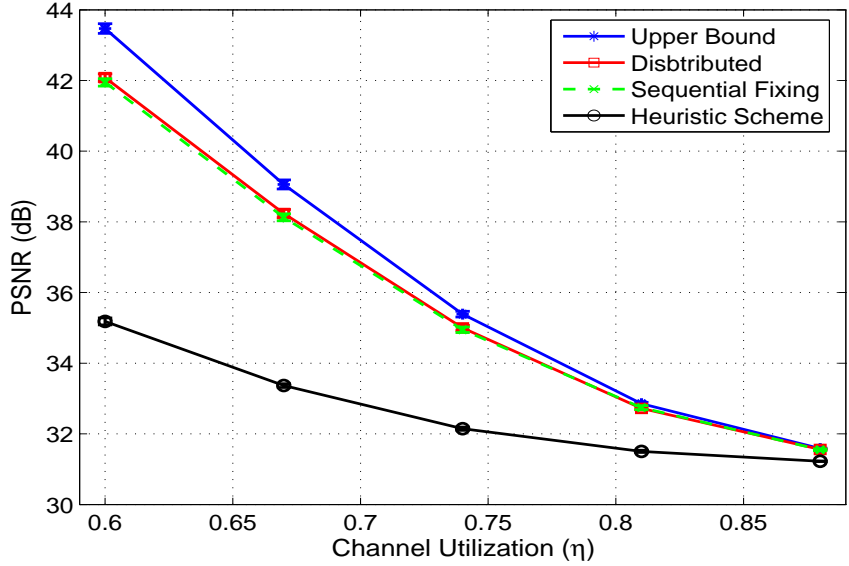}
  \caption{Video PSNRs versus primary user channel utilization $\eta$.}
  \label{fig:psnr-eta}
\end{figure}
%--------------------------------------------------------------

%--------------------------------------------------------------
\begin{figure}[!t]
  \centering  
  \includegraphics [width=4.5in, height=3.0in]{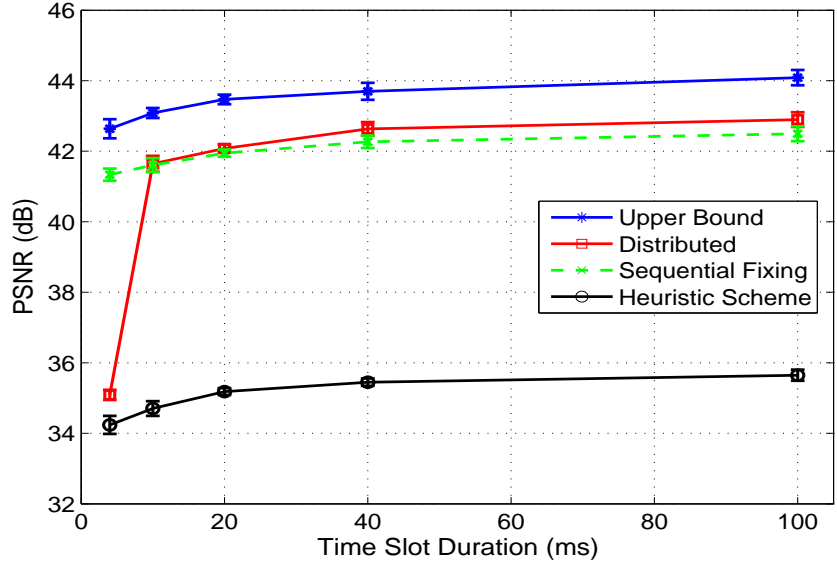}
  \caption{Impact of time slot duration on received video quality.}
  \label{fig:timeslot}
\end{figure}
%--------------------------------------------------------------

%--------------------------------------------------------------
\paragraph{Comparison of MPEG-4 FGS and H.264/SVC MGS Videos}

Finally, we compare MPEG-4 FGS and H.264/SVC MGS videos, while keeping the same settings.  It has been shown that H.264/SVC has better rate-distortion performance than MPEG-4 FGS due to the use of efficient hierarchical prediction structures, the inter-layer prediction mechanisms, improved drift control mechanism, and the efficient coding scheme in H.264/AVC~\cite{Wien07}.  Although MGS has Network Abstraction Layer (NAL) unit-based granularity, it achieves similar rate-distortion performance as H.264/SVC FGS~\cite{Wien07}. 

We plot the upper bounds and the distributed algorithm results in Figs.~\ref{fig:fgs-mgs-u} and~\ref{fig:fgs-mgs-error} for various channel utilizations and false alarm probabilities, respectively.  From the figures, it can be observed that there is a gap about 2.5 dB between the H.264/SVC MGS and MPEG-4 FGS curves, which clearly demonstrates the rate-distortion efficiency of MGS over MPEG-4 FGS.  The proposed algorithm can effectively handle both MGS and FGS videos, and the same trend is observed in both cases.

%-------------------------------------------------------------
\begin{figure}[!t]
  \centering  
  \includegraphics [width=4.5in, height=3.0in]{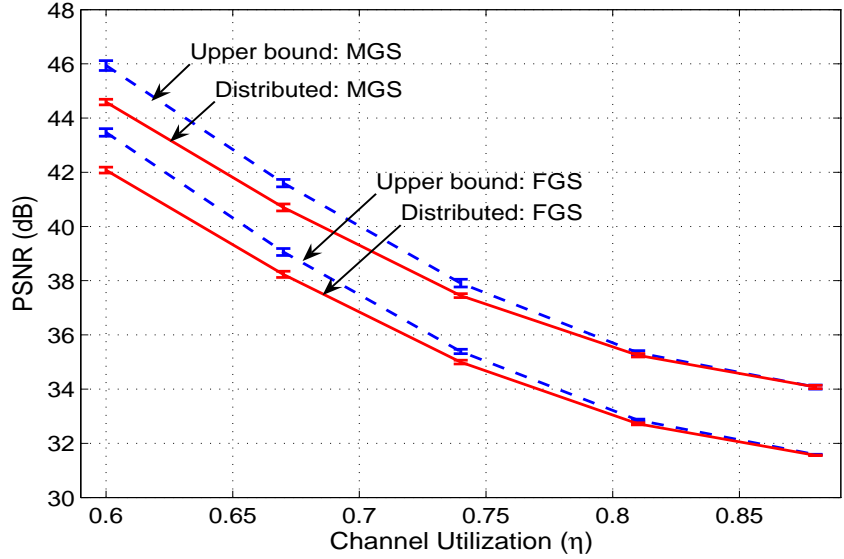}
  \caption{Comparison of MPEG-4 FGS video with H.264/SVC MGS video under various channel utilizations.}
  \label{fig:fgs-mgs-u}
\end{figure}
%-------------------------------------------------------------

%-------------------------------------------------------------
\begin{figure}[!t]
  \centering  
  \includegraphics [width=4.5in, height=3.0in]{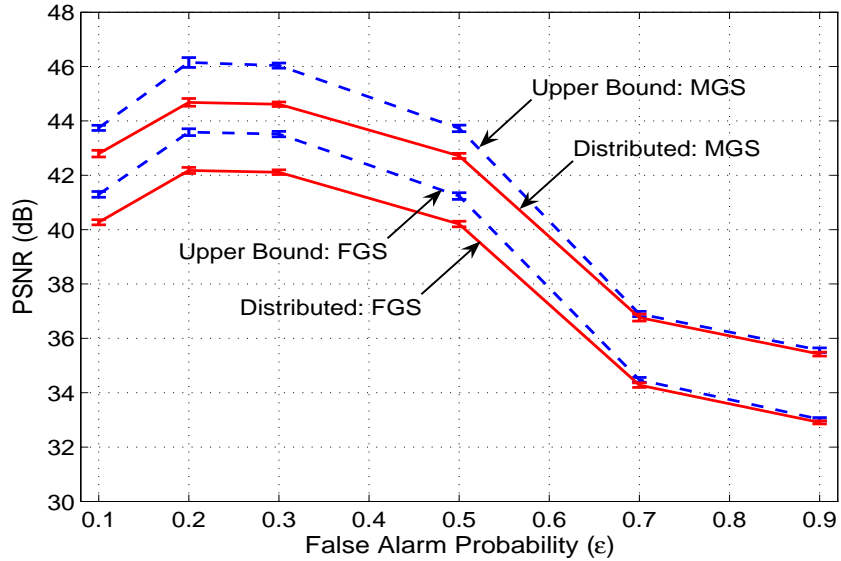}
  \caption{Comparison of MPEG-4 FGS video with H.264/SVC MGS video under various false alarm probabilities.}
  \label{fig:fgs-mgs-error}
\end{figure}
%-------------------------------------------------------------

%-------------------------------------------------------------------
% Conclusion
\section{Conclusions}\label{sec:cr_video_conc}
In this paper, we first addressed the problem of multicasting FGS video in CR networks. The problem formulation took video quality and proportional fairness as objectives, while considering cross-layer design factors such as FGS coding, spectrum sensing, opportunistic spectrum access, primary user protection, scheduling, error control and modulation. We proposed efficient optimization and scheduling algorithms for highly competitive solutions, and proved the complexity and optimality bound of the proposed greedy algorithm. Our simulation results demonstrate not only the viability of video over CR networks, but also the efficacy of the proposed approach. 

Then, we studied the challenging problem of streaming multiple scalable videos in a multi-hop CR network. The problem formulation considered spectrum sensing and sensing errors, spectrum access and primary user protection, video quality and fairness, and channel/path selection for concurrent video sessions.  We first solved the formulated MINLP problem using a sequential fixing scheme that produces lower and upper bounds on the achievable video quality. We then applied dual decomposition to derive a distributed algorithm, and analyzed its optimality and convergence performance. Our simulations validated the efficacy of the proposed scheme.

%-------------------------------------------------------------------
\bibliographystyle{IEEEtran}
\bibliography{cr_video_femto,MyWork}
%-------------------------------------------------------------------

\end{document}